\documentclass[twocolumn,aps,pra,showpacs,preprintnumbers, superscriptaddress, amsmath,amssymb]{revtex4}
\usepackage{amstext,mathrsfs,amsthm,bm}
\usepackage{mathtools} 
\usepackage{stmaryrd} 
\usepackage{hyperref}

\DeclareFontFamily{OT1}{rsfs}{}
\DeclareFontShape{OT1}{rsfs}{m}{n}{<5> rsfs5 <7> rsfs7 <10> rsfs10}{}
\DeclareSymbolFont{mathrsfs}{OT1}{rsfs}{m}{n}
\DeclareSymbolFontAlphabet{\mathrsfs}{mathrsfs}
\newcommand*{\bra}[1]{\left\langle{#1}\right|} 
\newcommand*{\ket}[1]{\left| #1 \right\rangle}

\newcommand*{\commut}[1]{[ #1 ]}
\newcommand*{\Commut}[1]{\left[ #1 \right]}

\newcommand*{\acommut}[1]{\{ #1 \}}
\newcommand*{\Acommut}[1]{\left\{ #1 \right\}}

\newcommand*{\poisson}[1]{\llbracket #1 \rrbracket}
\newcommand*{\Poisson}[1]{\left\llbracket #1 \right\rrbracket}

\newtheorem{theorem}{Theorem}
\newtheorem{lemma}{Lemma}

\begin{document}

\author{Denys I. Bondar}
\email{dbondar@princeton.edu}
\affiliation{Department of Chemistry, Princeton University, Princeton, NJ 08544, USA}
\affiliation{Department of Physics and Astronomy, University of Waterloo, Waterloo, Ontario N2L 3G1, Canada} 

\author{Renan Cabrera}
\email{rcabrera@princeton.edu}
\affiliation{Department of Chemistry, Princeton University, Princeton, NJ 08544, USA} 

\author{Robert R. Lompay}
\email{rlompay@gmail.com}
\affiliation{Department of Theoretical Physics, Uzhgorod National University, Uzhgorod 88000, Ukraine}

\author{Misha Yu. Ivanov}
\email{m.ivanov@imperial.ac.uk}
\affiliation{Imperial College, London SW7 2BW, United Kingdom}

\author{Herschel A. Rabitz}
\email{hrabitz@princeton.edu}
\affiliation{Department of Chemistry, Princeton University, Princeton, NJ 08544, USA} 

\title{Operational Dynamic Modeling Transcending Quantum and Classical Mechanics} 

\begin{abstract}
We introduce a general and systematic theoretical framework for Operational Dynamic Modeling (ODM) by combining a kinematic description of a model with the evolution of the dynamical average values. The kinematics includes the algebra of the observables and their defined averages. The evolution of the average values is drawn in the form of Ehrenfest-like theorems. We show that ODM is capable of encompassing wide ranging dynamics from classical non-relativistic mechanics to quantum field theory. The generality of ODM should provide a basis for formulating novel theories.  
\\
\\
This paper is published in \href{http://prl.aps.org/abstract/PRL/v109/i19/e190403}{{\it Phys. Rev. Lett.} {\bf 109}, 190403 (2012)}
\end{abstract}

\date{\today}


\maketitle

{\it Introduction.} One primary goal in science is to construct models possessing predictive capability. This endeavor is usually achieved by trial and error, with a proposed model either subsequently revised or completely discarded if its predictions do not agree with experimental results. Generally such a process is slow, hence automatization has been attempted \cite{Schmidt2009, King2009}. 

In this Letter, we develop a universal and systematic theoretical framework for {\it Operational Dynamic Modeling} (ODM) based on the evolution of dynamical average values.  As an illustration of ODM's scope, we infer quantum, classical, and unified quantum-classical mechanics. In order to construct  a system's dynamical model, we first postulate an associated kinematic description consisting of two independent components: i) the definition of the observables' average, and ii) the algebra of the observables. ODM applied to observable data, given in the form of Ehrenfest-like theorems [see, e.g., Eq. (\ref{AveragedEquationsForXandP})], returns the dynamical model (see Fig. \ref{Fig_GeneralEhrenfestQuantization} for a graphical summary). The system's kinematic description can also be deduced from complementary experiments. For example, if the results of a sequential measurement depend on the measurements' order, then the algebra of observables must be non-commutative [see comments after Eqs. (\ref{Ch2_ClassicalCommutator}) and (\ref{XP_CommutationalRelation})].  Limited access to experiments capable of firmly establishing the kinematics does not preclude hypothesizing plausible kinematic descriptions. Some of these hypotheses may be rejected within ODM by revealing their incompatibility with observable dynamical data
\footnote{For example, we demonstrate in Ref. \cite{Bondar2011} that finite-dimensional quantum and classical mechanics are not compatible with the Ehrenfest theorems (\ref{EhrenfestThs}).
}. 

In the spirit of ODM, starting from the Ehrenfest theorems [Eq. (\ref{EhrenfestThs})], we will obtain the Schr\"{o}dinger equation if the momentum and coordinate operators obey the canonical commutation relation, and the classical Liouville equation if the momentum and coordinate operators commute. To establish a link between quantum and classical mechanics, we introduce a generalized algebra of observables, incorporating both quantum and classical kinematics, that ultimately leads to {\it a unified quantum-classical mechanics}. Most importantly, we will show that ODM is applicable to a wide range of physical models from non-relativistic classical mechanics to quantum field theories, thus making ODM an important tool for formulating future models. 

{\it Preparing Dynamical Data.} In the current work, we present the conceptual and theoretical framework of ODM putting aside issues of handling noise contaminated experimental data. Assume we have multiple copies of either a quantum or classical system  (without loss of generality we consider single-particle one-dimensional systems throughout). Suppose we can precisely measure different copies of the particle's coordinate $x$ and momentum $p$ at times $\{ t_k \}_{k=1}^K$.  Upon performing ideal measurements of the coordinate or momentum on the $n$-th copy, we experimentally obtain $\{ x_n(t_k) \}$
 and $\{ p_n(t_k) \}$, $n=1,\ldots,N$, requiring a total of $2KN$ observations. Time interpolation of these data points returns the functions $x_n(t)$ and $p_n(t)$. We may then calculate the  statistical moments 
$ \overline{ [x(t)]^l } = \frac{1}{N} \sum_{n=1}^N [x_n (t)]^l$ and
$\overline{ [p(t)]^l } = \frac{1}{N} \sum_{n=1}^N [p_n (t)]^l$ for $l=1,2,3,\ldots$  We make the {\it ansatz}, resembling a Taylor series with coefficients $a_l$, $b_l$, $c_{k,l}$, $d_l$, $e_l$, and $f_{k,l}$, that the first derivative of $\overline{x(t)} = \overline{ [x(t)]^1 }$ and $\overline{p(t)} = \overline{ [p(t)]^1 }$ satisfy
\begin{align}
	\frac{d}{dt} \overline{x(t)} &= \sum_l \left( a_l \overline{[x(t)]^l} + b_l \overline{[p(t)]^l} \right)
		+ \sum_{k,l \neq 0} c_{k,l} \overline{[x(t)]^l [p(t)]^k}, \notag\\
	\frac{d}{dt} \overline{p(t)} &= \sum_l \left( d_l \overline{[x(t)]^l} + e_l \overline{[p(t)]^l} \right)
		+ \sum_{k,l \neq 0} f_{k,l} \overline{[x(t)]^l [p(t)]^k}. \notag
\end{align}
For non-dissipative quantum and classical systems, these relations reduce to
\begin{align}\label{AveragedEquationsForXandP}
	m\frac{d}{dt} \overline{x(t)} = \overline{p(t)}, \qquad \frac{d}{dt} \overline{p(t)} = \overline{-U'(x)}(t),
\end{align}
where $\overline{-U'(x)}(t) = \sum_l d_l  \overline{[x(t)]^l}$.

{\it Kinematic Description.} Generalizing Schwinger's motto ``quantum mechanics: symbolism of atomic measurements'' \cite{Schwinger2003}, we adapt that any physical model is a symbolic representation of the experimental evidence supporting it. The mathematical symbolism for this purpose needs to be considered. A formalism specialized to describe a specific class of behavior (e.g., classical mechanics expressed in terms of phase space trajectories) can be effective, but it may be unsuitable for connecting different classes of phenomena (e.g.,  unifying quantum and classical mechanics). In this case a general and versatile formalism is preferred. Building a formalism around Hilbert space is a suitable candidate for this role. Hilbert space is well understood, rich in mathematical structure, and convenient for practical computations. 	

Consider the postulates: i) The states of a system are represented by normalized vectors $\ket{\Psi}$ of a complex Hilbert space, and the observables are given by self-adjoint operators acting on this space; ii) The expectation value of a measurable $\hat{A}$ at time $t$ is $\overline{A}(t) =  \bra{\Psi (t)} \hat{A} \ket{\Psi(t)}$; iii) The probability that a measurement of an observable $\hat{A}$ at time $t$ yields $A$ is $\left|\langle A \ket{\Psi(t)}\right|^2$, where $\hat{A} \ket{A} = A \ket{A}$; iv) The state space of a composite system is the tensor product of the subsystems' state spaces. Having accepted these postulates, the rest -- state spaces, observables, and the equations of motion -- can be deduced directly from  observable data. Importantly, these axioms are just the well-known quantum mechanical postulates with the adjective ``quantum'' {\it removed, as $\ket{\Psi}$ is a general state encompassing classical and quantum behavior.} We will demonstrate below that these postulates are sufficient to capture all the features of both quantum {\it and} classical mechanics as well as the associated hybrid mechanics. Equation (\ref{AveragedEquationsForXandP}) rewritten in terms of the axioms becomes
\begin{align}\label{EhrenfestThs}
	m\frac{d}{dt} \bra{\Psi(t)} \hat{x} \ket{\Psi(t)} & = \bra{\Psi(t)} \hat{p} \ket{\Psi(t)}, \notag\\
	\frac{d}{dt} \bra{\Psi(t)} \hat{p} \ket{\Psi(t)} & = \bra{\Psi(t)} -U'(\hat{x}) \ket{\Psi(t)}.
\end{align}

Koopman and von Neumann \cite{Koopman1931, Neumann1932} pioneered the recasting of classical mechanics in a form similar to quantum mechanics by introducing classical complex valued wave functions and representing associated physical observables by means of commuting self-adjoint operators (for modern developments and applications see Refs. \cite{Gozzi1988, Gozzi1989, Wilkie1997, Wilkie1997a, Gozzi2002, DaniloMauro2002, Deotto2003, Deotto2003a, Abrikosovjr2005, Blasone2005, Brumer2006, Carta2006, Gozzi2010, Gozzi2011, Cattaruzza2011}). Our operational formulation is closely related to the approach proposed in Ref. \cite{Bogdanov2004} and recently successfully implemented for quantum state tomography \cite{Bogdanov2010, Bogdanov2011}. Regarding developments of other operational approaches see Ref. \cite{Hardy2011} and references therein.

{\it Inference of Classical Dynamics.} Let $\hat{x}$ and $\hat{p}$ be self-adjoint operators representing the coordinate and momentum observables. The commutation relationship 
\begin{align}\label{Ch2_ClassicalCommutator}
	\commut{\hat{x}, \hat{p}} = 0,
\end{align}
encapsulates two basic experimental facts of classical kinematics: i) the position and momentum can be measured simultaneously with arbitrary accuracy, ii) observed values do not depend on the order of performing the measurements. In terms of our axioms, the dynamical observations of the classical particle's position and momentum are summarized in Eq. (\ref{EhrenfestThs}).

We now derive the equation of motion for a classical state. The application of the chain rule to Eq. (\ref{EhrenfestThs}) gives
\begin{align}\label{Ch2_ExpandedClassicsErhenfest_Th}
	\bra{d\Psi/dt} \hat{x} \ket{\Psi} + \bra{\Psi} \hat{x} \ket{d\Psi/dt} &= \bra{\Psi} \hat{p}/m \ket{\Psi}, \notag\\
	\bra{d\Psi/dt} \hat{p} \ket{\Psi}  + \bra{\Psi} \hat{p} \ket{d\Psi/dt} & = \bra{\Psi} -U'(\hat{x}) \ket{\Psi},
\end{align}
into which we substitute a consequence of Stone's theorem (see Sec. I \ref{Sec_Stones_Th})
\begin{align}\label{Ch2_AlmostLiouville_Eq}
	i \ket{d\Psi(t)/dt} = \hat{L} \ket{\Psi(t)},
\end{align}
and obtain
\begin{align}\label{Avaraged_Equations_for_Liouvillian}
	im \bra{\Psi(t)} \commut{\hat{L}, \hat{x} } \ket{\Psi(t)} &= \bra{\Psi(t)} \hat{p} \ket{\Psi(t)}, \notag\\
	i \bra{\Psi(t)} \commut{\hat{L}, \hat{p}} \ket{\Psi(t)} &= - \bra{\Psi(t)} U'(\hat{x}) \ket{\Psi(t)}.
\end{align}
Since Eq. (\ref{Avaraged_Equations_for_Liouvillian}) must be valid for all possible initial states, the averaging can be dropped, and we have the  system of commutator equations for the motion generator $\hat{L}$,
\begin{align}\label{Ch2_Equations_for_Liouvillian}
	im  \commut{\hat{L}, \hat{x}}  =  \hat{p} , \qquad i \commut{\hat{L}, \hat{p}} = -U'(\hat{x}).
\end{align}
Since $\hat{p}$ and $\hat{x}$ commute, the solution $\hat{L}$ cannot be found by simply assuming $\hat{L} = L(\hat{x}, \hat{p})$ (regarding the definition of functions of operators see Sec. \ref{Sec_Weyl_calculus}). We add into consideration two new operators $\hat{\lambda}_x$ and $\hat{\lambda}_p$ such that
\begin{align}\label{Complete_classical_algebra}
	\commut{ \hat{x}, \hat{\lambda}_x } = \commut{ \hat{p}, \hat{\lambda}_p } = i, 
\end{align}
and the other commutators among $\hat{x}$, $\hat{p}$, $\hat{\lambda}_x$, and $\hat{\lambda}_p$ vanish. The need to introduce auxiliary operators arises in classical dynamics because all the observables commute;  hence, the notion of an individual trajectory can be introduced (see also Sec. \ref{Sec_Classical_FieldTh}). Moreover, the choice of the commutation relationships (\ref{Complete_classical_algebra}) is unique. Equation (\ref{Complete_classical_algebra}) can be considered as an additional axiom. Now we seek the generator $\hat{L}$ in the form $\hat{L} = L(\hat{x}, \hat{\lambda}_x, \hat{p}, \hat{\lambda}_p)$. Utilizing Theorem \ref{Th_Weyl_commutator_theorem} from Sec. \ref{Sec_Weyl_calculus}, we convert the commutator equations (\ref{Ch2_Equations_for_Liouvillian}) into the differential equations
\begin{align}
	m L'_{\lambda_x} (x, \lambda_x, p, \lambda_p) = p, \, L'_{\lambda_p} (x, \lambda_x, p, \lambda_p) = -U'(x),
\end{align}
from which, the generator of classical dynamics $\hat{L}$ is found to be
\begin{align}\label{General_form_of_Liouvillian}
	\hat{L} = \hat{p} \hat{\lambda}_x / m - U'(\hat{x}) \hat{\lambda}_p + f(\hat{x}, \hat{p}),
\end{align}
where $f(x,p)$ is an arbitrary real-valued function.
Equations (\ref{Ch2_AlmostLiouville_Eq}), (\ref{Complete_classical_algebra}), and (\ref{General_form_of_Liouvillian}) represent classical dynamics in an abstract form. 

Let us find the equation of motion for $|\langle p \, x \ket{\Psi(t)}|^2$ by rewriting Eq. (\ref{Ch2_AlmostLiouville_Eq}) in the $xp$-representation (in which $\hat{x} = x$, $\hat{\lambda}_x = -i\partial / \partial x$, $\hat{p} = p$, and $\hat{\lambda}_p = -i \partial / \partial p$),
\begin{align}\label{PreLiouvillianEq}
		\left[ i\frac{\partial }{\partial t} +i\frac{p}{m} \frac{\partial}{\partial x} - i U'(x) \frac{\partial}{\partial p} - f(x,p) \right]  \langle p \, x \ket{\Psi(t)} = 0,
\end{align}
which yields the well known classical Liouville equation for the probability distribution in phase-space $\rho(x,p;t) = |\langle p \, x \ket{\Psi(t)}|^2$,
\begin{align}\label{Liouville_Eq}
	\frac{\partial }{\partial t} \rho(x,p;t)  = \left[ -\frac{p}{m} \frac{\partial}{\partial x} +  U'(x) \frac{\partial}{\partial p}  \right]  \rho(x,p;t).
\end{align}
Thus, we have deduced the classical Liouville equation along with the Koopman-von Neumann theory from Eq. (\ref{EhrenfestThs}) by assuming that the classical momentum and coordinate operators commute.

{\it Inference of Quantum Dynamics.} The hallmark of quantum kinematics is the canonical commutation relation
\begin{equation}\label{XP_CommutationalRelation}
	 \commut{ \hat{x}, \hat{p} } = i\hbar,
\end{equation}
which implies i) the Heisenberg uncertainty principle and ii) the order of performing measurements of the coordinate and momentum does matter \cite{Schwinger2003}. The evolution of expectation values of the quantum coordinate and momentum is governed by the Ehrenfest theorems (\ref{EhrenfestThs}). 

We repeat the algorithm exercised in classical mechanics above. Substituting the definition of the motion generator $\hat{H}$ obtained from Stone's theorem (see Sec. \ref{Sec_Stones_Th})
\begin{align}\label{AbstractSchrodingerEq}
	i\hbar \ket{d \Psi(t)/dt} = \hat{H} \ket{\Psi(t)},
\end{align}
into Eq. (\ref{EhrenfestThs}), we obtain
\begin{align}\label{Commut_Eq_for_Schrod}
	im \commut{\hat{H}, \hat{x}} = \hbar \hat{p}, \qquad i\commut{\hat{H}, \hat{p}} = -\hbar U'(\hat{x}).
\end{align}
Assuming $\hat{H} = H(\hat{x}, \hat{p})$ and utilizing Theorem \ref{Th_Weyl_commutator_theorem} from Sec. \ref{Sec_Weyl_calculus}, the commutation relations in Eq. (\ref{Commut_Eq_for_Schrod}) reduce to $m H'_p (x,p) = p$ and $H'_x (x,p) = U'(x)$. Whence, the familiar quantum Hamiltonian readily follows
\begin{align}\label{Quantum_Hamiltonian}
	\hat{H} = \hat{p}^2 /(2m) + U(\hat{x}).
\end{align} 
Since the Schr\"{o}dinger equation was derived from the Ehrenfest theorems (\ref{EhrenfestThs}) assuming the canonical commutation relation (\ref{XP_CommutationalRelation}), the presentation suggests that the Ehrenfest theorems are more fundamental than the Schr\"{o}dinger equation.

{\it Unification of Quantum and Classical Mechanics.} (For a detailed discussion see Sec. \ref{Sec_Qunatum_to_Calssical}; see also Fig. \ref{Fig_EhrenfestQuantization}) The fundamental difference between non-relativistic classical and quantum mechanics is that the momentum and coordinate operators commute in the former case and do not commute in the latter \cite{Dirac1958, Shirokov1976b, Shirokov1979d}. The operators $\hat{x}$, $\hat{p}$, $\hat{\lambda}_x$, and $\hat{\lambda}_p$ obeying Eq. (\ref{Complete_classical_algebra}) form the classical operator algebra. The unified quantum-classical operator algebra is based on $\hat{x}_q$, $\hat{p}_q$, $\hat{\vartheta}_x$, and $\hat{\vartheta}_p$ satisfying
\begin{align}
	\commut{ \hat{x}_q, \hat{p}_q } = i\hbar\kappa, \qquad
	\commut{ \hat{x}_q, \hat{\vartheta}_x } = \commut{ \hat{p}_q, \hat{\vartheta}_p } = i,
\end{align}
$0 \leqslant \kappa \leqslant 1$, while all the other commutators among $\hat{x}_q$, $\hat{p}_q$, $\hat{\vartheta}_x$, and $\hat{\vartheta}_p$ vanish. The operators $\hat{\vartheta}_x$ and $\hat{\vartheta}_p$ are simply introduced so that the quantum algebra (i.e., $\kappa = 1$) is consistent with the classical algebra. The limit $\kappa \to 0$ defines the quantum-to-classical transition with the quantum algebra smoothly transforming into the classical one as $\kappa \to 0$. Since $\hbar$ enters in the time derivative of Sch\"{o}dinger equation (\ref{AbstractSchrodingerEq}) as well as in the commutator relationship (\ref{XP_CommutationalRelation}), the limit $\hbar\to 0$ encompasses more than the criterion that the coordinate and momentum operators must commute in the classical limit. This situation motivated the introduction of the parameter $\kappa$.

As the first step towards unification of both mechanics, we apply ODM to 
\begin{align}\label{EhrenfestThrs_in_quantumclass_case}
	m\frac{d}{dt} \bra{\Psi(t)} \hat{x}_q \ket{\Psi(t)} &= \bra{\Psi(t)} \hat{p}_q \ket{\Psi(t)}, \notag\\
	\frac{d}{dt} \bra{\Psi(t)} \hat{p}_q \ket{\Psi(t)} &= \bra{\Psi(t)} -U'(\hat{x}_q) \ket{\Psi(t)},
\end{align}
and obtain the Hamiltonian 
\begin{align}\label{MixedQuantumClassical_Hamiltonian}
	\hat{\mathcal{H}} = \frac{1}{\kappa} \left[ \frac{ \hat{p}_q^2 }{2 m} + U(\hat{x}_q) \right] 
					+ F \left( \hat{p}_q - \hbar\kappa \hat{\vartheta}_x, \hat{x}_q + \hbar\kappa \hat{\vartheta}_p \right),
\end{align}
such that  $i\hbar \ket{d \Psi(t)/dt} = \hat{\mathcal{H}} \ket{\Psi(t)}$, where $F$ is an arbitrary real-valued smooth function. Note that no  Ehrenfest theorems for the observables $\hat{O} = O\left( \hat{x}_q, \hat{p}_q \right)$ can specify the function $F$ because $\commut{\hat{F}, \hat{O}} = 0$. Hence, the function $F$ is experimentally undetectable. We shall utilize this freedom by finding an $F$ which enforces that the Hamiltonian (\ref{MixedQuantumClassical_Hamiltonian}) smoothly transform to become the Liouvillian (\ref{General_form_of_Liouvillian}) in the classical limit.

The classical and quantum algebras are isomorphic. The quantum operators can be constructed as linear combinations of the classical operators in many ways, e.g.,
\begin{align}\label{ClassicQuantumAlg_isomorphism}
	\hat{x}_q = \hat{x} - \hbar\kappa \hat{\lambda}_p / 2, \qquad & \hat{p}_q = \hat{p} + \hbar\kappa \hat{\lambda}_x / 2, \notag\\
	\hat{\vartheta}_x = \hat{\lambda}_x, \qquad & \hat{\vartheta}_p = \hat{\lambda}_p.
\end{align}
In particular, demanding that the quantum operators are expressed as linear combinations of the classical ones such that 
\begin{align}\label{Conditions_Specification_of_F_Th1}
	& \lim_{\kappa \to 0} \hat{x}_q = \hat{x}, \qquad
		 \lim_{\kappa \to 0} \hat{p}_q = \hat{p}, \qquad
		 \lim_{\kappa \to 0} \hat{\theta}_x = \hat{\lambda}_x, \notag\\
	& \lim_{\kappa \to 0} \hat{\theta}_p = \hat{\lambda}_p, \qquad
		 \lim_{\kappa \to 0} \hat{\mathcal{H}} = \hbar \hat{L},
\end{align}
identifies the function $F$ as (see Theorems \ref{Specification_of_F_Th1} and \ref{Specification_of_F_Th2} in Sec. \ref{Sec_Qunatum_to_Calssical})
\begin{align}\label{Specification_of_F_Th1_Eqresult}
	F(p, x) = - p^2 /(2 m \kappa) - U(x)/\kappa + O(1). \quad (\kappa \to 0) 
\end{align}
Keeping the leading term in Eq. (\ref{Specification_of_F_Th1_Eqresult}), we show in Sec. \ref{Sec_Qunatum_to_Calssical} that only
 isomorphism (\ref{ClassicQuantumAlg_isomorphism}) is compatible with such a function $F$, which leads to the final expression for the unified quantum-classical Hamiltonian, 
\begin{align}\label{Hamiltonian_Hqc}
	\hat{\mathcal{H}}_{qc} &= \frac{1}{\kappa} \left[ \frac{ \hat{p}_q^2 }{2 m} + U(\hat{x}_q) \right] 
						- \frac{1}{2m\kappa} \left( \hat{p}_q - \hbar\kappa \hat{\vartheta}_x \right)^2 \notag\\
					&\qquad - \frac{1}{\kappa} U \left( \hat{x}_q + \hbar\kappa \hat{\vartheta}_p \right) \notag\\
					& \equiv \frac{\hbar}{m} \hat{p} \hat{\lambda}_x + 
						\frac{1}{\kappa} U\left( \hat{x} - \frac{\hbar\kappa}{2} \hat{\lambda}_p \right) -
						\frac{1}{\kappa} U\left( \hat{x} + \frac{\hbar\kappa}{2} \hat{\lambda}_p \right).
\end{align}  
that fulfills conditions (\ref{Conditions_Specification_of_F_Th1}). 
 Theorem \ref{Golden_Harmonic_Osc_Theorem} in Sec. \ref{Sec_Qunatum_to_Calssical} states that $\hat{\mathcal{H}}_{qc} \equiv \hbar \hat{L}$ for any value of $\kappa$ if and only if $U$ is a quadratic polynomial. 

We now demonstrate that the Wigner phase-space representation is a special case of the unified mechanics. First rewriting the equation of motion
\begin{align}\label{SchrodEq_with_Hqc}
	i \hbar \ket{ d\Psi_{\kappa}(t) / dt } = \hat{\mathcal{H}}_{qc} \ket{\Psi_{\kappa}(t)}
\end{align}
in the $x\lambda_p$-representation (for which $\hat{x} = x$, $\hat{\lambda}_x = -i \partial / \partial x$,  $\hat{p} = i \partial / \partial \lambda_p$, and $\hat{\lambda}_p = \lambda_p$), then introducing new variables $u = x - \hbar\kappa \lambda_p / 2$ and $v = x + \hbar\kappa \lambda_p / 2$, we transform Eq. (\ref{SchrodEq_with_Hqc}) into
\begin{align} 
	\left[ i \hbar\kappa \frac{\partial}{\partial t} - \frac{(\hbar\kappa)^2}{2m}\left( \frac{\partial^2}{\partial v^2} 
	- \frac{\partial^2}{\partial u^2} \right) - U(u) + U(v) \right] \rho_{\kappa} = 0, \notag
\end{align}
where $\rho_{\kappa}(u,v;t) \propto \langle x \, \lambda_p \ket{\Psi_{\kappa}(t)}$. Therefore, $\rho_{\kappa}$ is the density matrix for a quantum system with the Hamiltonian (\ref{Quantum_Hamiltonian}) after substituting $\hbar \to \hbar\kappa$. Note that $\kappa$ enters the equation of motion (\ref{SchrodEq_with_Hqc}) as only a multiplicative constant renormalizing $\hbar$. From this perspective, the limit $\kappa \to 0$ is indeed equivalent to $\hbar \to 0$. The transition from the $x\lambda_p$- to $xp$-representation results in
\begin{align}
	\langle p \, x \ket{\Psi_{\kappa}(t)} = \sqrt{\frac{\hbar\kappa}{2\pi}} \int  d\lambda_p 
		\rho_{\kappa} \left(x - \frac{\hbar\kappa \lambda_p} {2},  x + \frac{\hbar\kappa \lambda_p}{ 2}; t \right) e^{ip\lambda_p}.
\end{align}
Hence, the wave function $\langle p \, x \ket{\Psi_{\kappa}(t)}$ is proportional to the celebrated Wigner quasi-probability distribution. 

By only demanding a consistent melding of quantum and classical mechanics within ODM, we achieved the construction equivalent to the Wigner phase-space formulation of quantum mechanics. The great attraction of the Wigner formalism is due to its smooth and physically consistent quantum-to-classical and classical-to-quantum transitions \cite{Wigner1932, Heller1976, Shirokov1976b, Shirokov1979d, Kapral1999, Bolivar2004, Zachos2005, Kapral2006}. Our analysis also points to a unique feature of the phase-space formulation: no quantum mechanical representation, but Wigner's, has a ``nice'' classical limit. Moreover, since the Wigner function's dynamical equation is recast in the form of a Schr\"{o}dinger-like equation (\ref{SchrodEq_with_Hqc}), 
efficient numerical methods for solving the Schr\"{o}dinger equation may be applied to propagate the Wigner function for conceptual appeal and practical utility.

{\it Future Prospects.} 
 ODM was introduced to derive equations of motion from the evolution of average values and a chosen kinematical description. In Secs. \ref{Sec_CanonicalQuantz}-\ref{Sec_Quantum_FieldTh}, ODM is applied to the canonical quantization rule, the Schwinger quantum action principle, the time measuring problem in quantum mechanics, quantization in curvilinear coordinates, as well as classical and quantum field theories. Additionally, relativistic classical and quantum mechanics is also melded within this framework in Ref. \cite{Cabrera2011}.

Variational principles are at the heart of physics. Within their framework, the problem of model generation is reduced to finding the correct form of the action functional, whose Euler-Lagrange equations govern the model's dynamics. However, the action is usually neither directly observable nor unique; hence, its construction is a subject of debate and can only be justified {\it post factum} by supplying experimentally verifiable equations of motion. More important, there are phenomena beyond the scope of variational principles (e.g., dissipation). ODM is a theoretical framework free of all these conceptual weaknesses since it operates with observable data recast in the form of Ehrenfest-like relations. Hence, the equations of motion are no longer axioms but are corollaries of the more fundamental Ehrenfest theorems.

{\it Acknowledgments.}
D.I.B., R.C., and H.A.R. acknowledge support from NSF and ARO. Fruitful discussions with Dmitry Zhdanov are much appreciated.

\begin{widetext}

\begin{center}
{\Large\bf Supplemental Material for: ``Operational Dynamic Modeling Transcending Quantum and Classical Mechanics''} 
\end{center}

In the main text of the Letter, we introduced Operational Dynamic Modeling (ODM) as a consistent universal theoretical framework for inferring dynamical models from observable data (idealized in this work as noise free). To construct a system's model, ODM requires: i) the definition of observables' averaging, ii) the algebra of the observables, and iii) observable evolution of the average values (see Fig. \ref{Fig_GeneralEhrenfestQuantization}). The purpose of this supplemental material is to employ this technique to encompass a variety of dynamical models not covered in the main text (see the list below). Additionally, we provide a detailed derivation of the unified mechanics in Sec. \ref{Sec_Qunatum_to_Calssical} (see Fig. \ref{Fig_EhrenfestQuantization} for the roadmap of this derivation).

\begin{figure*}
	\begin{center}
		\includegraphics[scale=0.40]{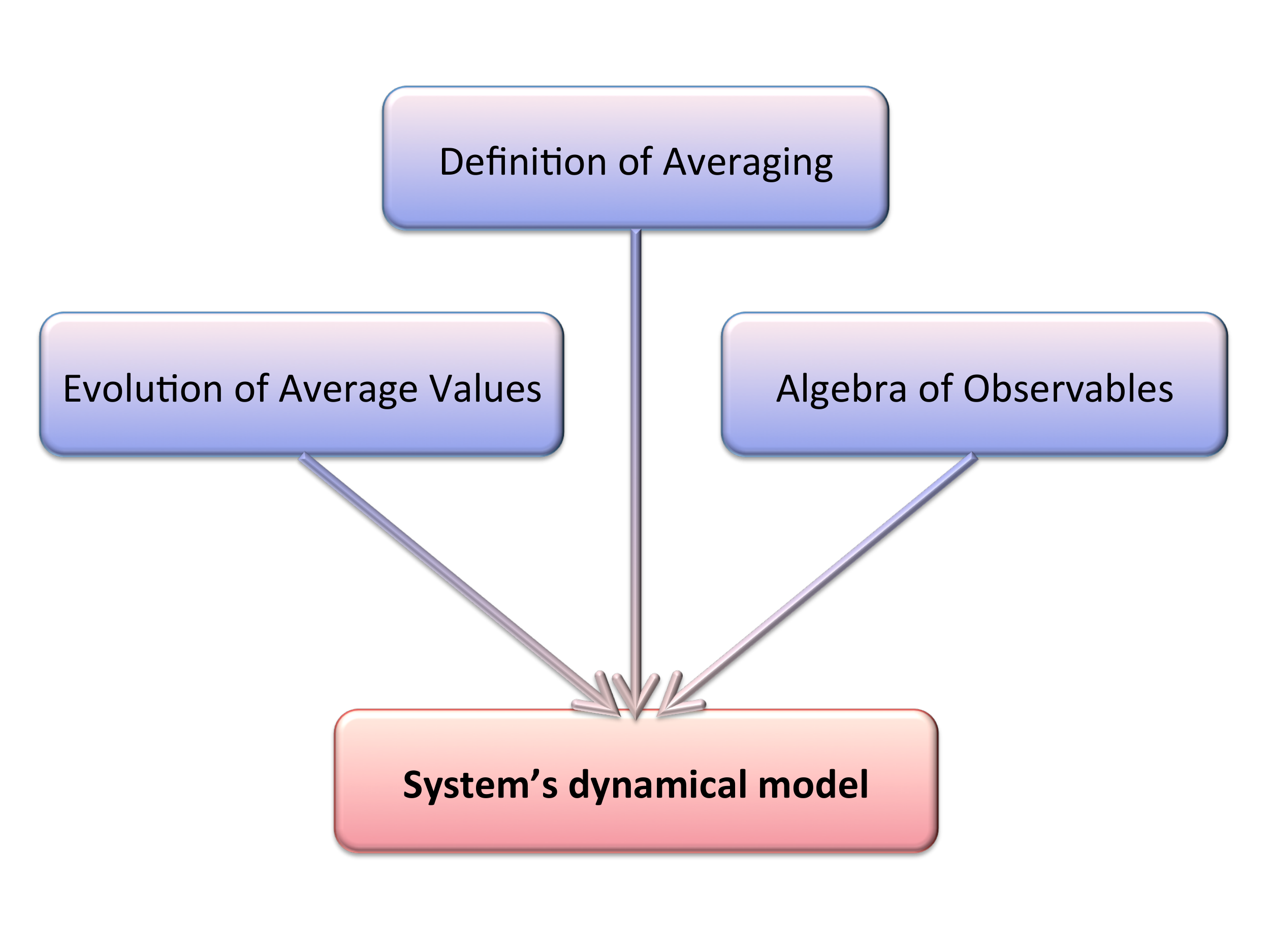}
		\caption{The general scheme of Operational Dynamic Modeling.}\label{Fig_GeneralEhrenfestQuantization}
	\end{center}
\end{figure*}
\begin{figure*}
	\begin{center}
		\includegraphics[scale=0.40]{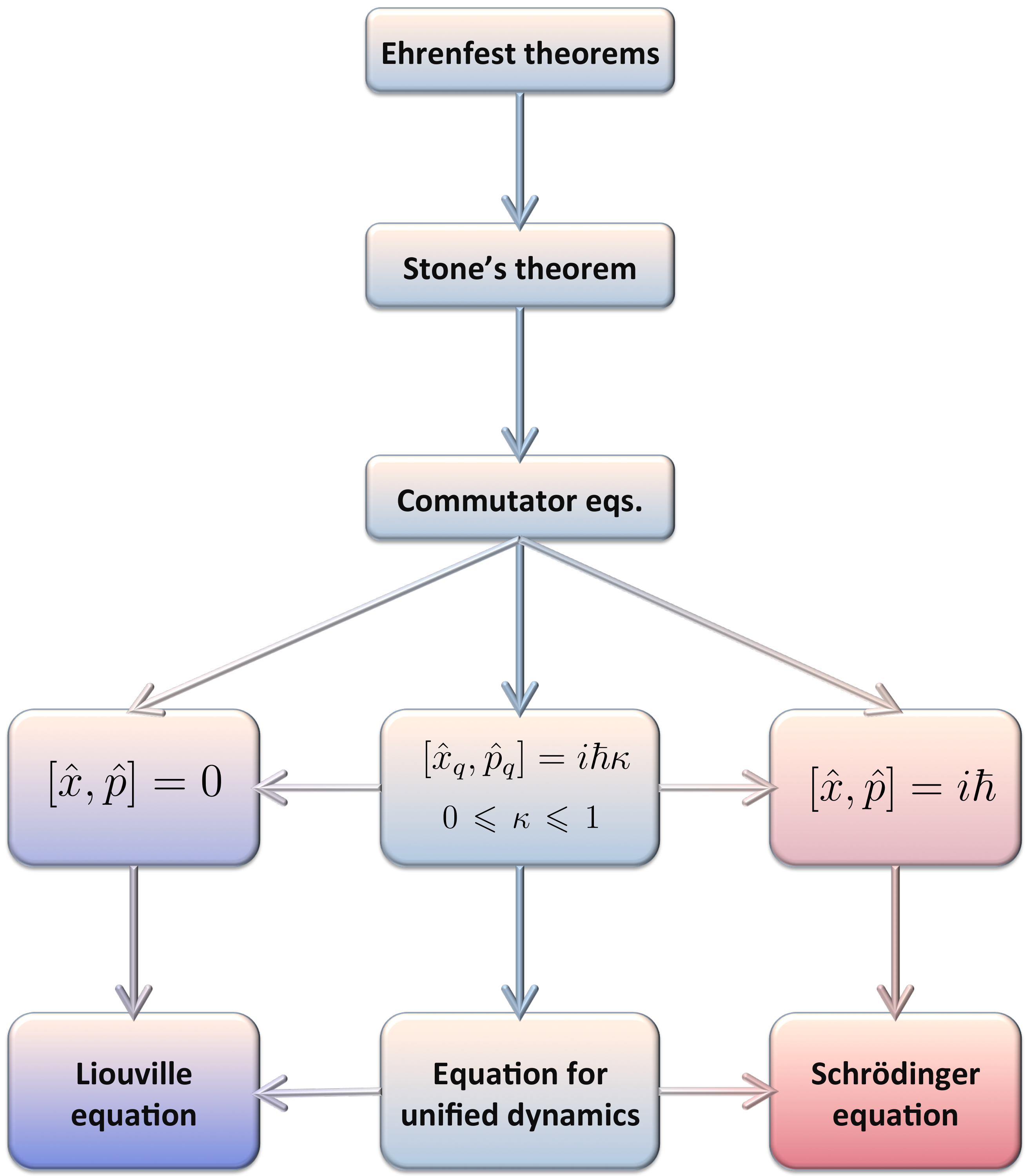}
		\caption{The derivation of quantum, classical, and unified mechanics within Operational Dynamic Modeling.}\label{Fig_EhrenfestQuantization}
	\end{center}
\end{figure*}

Before proceeding further, we wish to clarify a few points. There is a widespread belief that the Ehrenfest theorems cannot shed light on the quantum-to-classical transition. Such claims are partially due to terminology. Ehrenfest \cite{Ehrenfest1927} derived the following
\begin{align}\label{EhrenfestThs__}
	m\frac{d}{dt} \bra{\Psi(t)} \hat{x} \ket{\Psi(t)} = \bra{\Psi(t)} \hat{p} \ket{\Psi(t)}, \qquad
	\frac{d}{dt} \bra{\Psi(t)} \hat{p} \ket{\Psi(t)} = \bra{\Psi(t)} -U'(\hat{x}) \ket{\Psi(t)},
\end{align}
which we will exclusively refer to as {\it ``the Ehrenfest theorems''}. However, the same label is often applied to mean that the centroid of a narrow wave-packet follows a classical trajectory, i.e., 
\begin{align}\label{Pseudo_EhrenfestThs}
	m\frac{d}{dt} \bra{\Psi(t)} \hat{x} \ket{\Psi(t)}  = \bra{\Psi(t)} \hat{p} \ket{\Psi(t)}, \qquad
	\frac{d}{dt} \bra{\Psi(t)} \hat{p} \ket{\Psi(t)}  \approx  -U' \left( \bra{\Psi(t)} \hat{x} \ket{\Psi(t)} \right).
\end{align}
While Eq. (\ref{EhrenfestThs__}) is a rigorous mathematical identity \cite{Friesecke2009, Friesecke2010}, Eq. (\ref{Pseudo_EhrenfestThs}) is an assertion based on a physical approximation \cite{Ballentine1994, Ballentine2004}.

\tableofcontents

\section{Stone's Theorem}\label{Sec_Stones_Th}

Stone's theorem \cite{Reed1980, Araujo2008} can be stated as follows: If $\hat{U}(t)$ is a strongly continuous unitary group (e.g., which describes the evolution of a system), then there is a unique self-adjoint operator $\hat{H}$ (the dynamic generator, e.g., Hamiltonian or Liouvillian) such that
\begin{align}
	i \frac{d}{d t}\hat{U}(t) \ket{ f }= \hat{H}\hat{U}(t) \ket{ f }, 
\end{align}
for all $\ket{f}$ from the domain of the operator $\hat{H}$. The latter equation can also be formally expressed as $\hat{U}(t) = \exp(-i\hat{H}t)$. 

If the time-independent generator of motion $\hat{H}$ is a self-adjoint operator, then Stone's theorem  guarantees not only the existence of unique solutions of the time-independent Schr\"{o}dinger and Liouville equations, but also the conservation of the wave function norms. Regarding the generalization of Stone's theorem to the case of a time-dependent Hamiltonian see, e.g., Sec. X.12 of Ref. \cite{Reed1975}.

Physically, Stone's theorem is equivalent to assuming that observables smoothly depend on time.

\section{Noncommutative Analysis: The Weyl Calculus}\label{Sec_Weyl_calculus}

Noncommutative analysis \cite{Weyl1950, Feynman1951, Nelson1970, Taylor1973, Maslov1976, Karasev1979, Nazaikinskii1992, Nazaikinskii1996, Madore2000, Andersson2004, Jefferies2004, Rosas2004, Nielsen2005, Muller2007} is a broad and active field of mathematics with a number of important applications. This branch of analysis aims at identifying functions of noncommutative variables and specifying operations with such objects. There are many ways of introducing functions of operators; however, the choice of a particular definition is a matter of convenience \cite{Nazaikinskii1992}. 

To make the paper self-consistent, we shall review basic results from the Weyl calculus, which is a popular version of noncommuting analysis. Theorem \ref{Th_Weyl_commutator_theorem} plays a crucial role in the current paper. Even though we prove this result within the Weyl calculus, it is valid in more general settings (see, e.g., Ref. \cite{Maslov1976} and page 63 of Ref. \cite{Nazaikinskii1996}).

The starting point is the well known fact that Fourier transforming back and forth  does not change a sufficiently smooth function of $n$-arguments,
\begin{align}\label{Fourier_identity}
	f(\lambda_1, \ldots, \lambda_n) = \frac{1}{(2\pi)^n} \int \prod_{l=1}^n d\xi_l d\eta_l 
		\exp\left[ i \sum_{q=1}^n \eta_q(\lambda_q - \xi_q) \right] f(\xi_1, \ldots, \xi_n).
\end{align}
Following this observation, we define the function of noncommuting operators within the Weyl calculus as 
\begin{align}\label{Weyl_main_definition}
	f(\hat{A}_1, \ldots, \hat{A}_n) \coloneqq \frac{1}{(2\pi)^n} \int \prod_{l=1}^n d\xi_l d\eta_l
		\exp\left[ i \sum_{q=1}^n \eta_q(\hat{A}_q - \xi_q) \right] f(\xi_1, \ldots, \xi_n),
\end{align}
where the exponential of an operator is specified by the Taylor expansion,
\begin{align}\label{OperatorExp}
	\exp(\hat{A}) \coloneqq \sum_{k=0}^{\infty} \frac{\hat{A}^k}{ k! }.
\end{align}
The identity
\begin{align}\nonumber
	f^{\dagger} (\hat{A}_1, \ldots, \hat{A}_n) = f(\hat{A}_1^{\dagger}, \ldots, \hat{A}_n^{\dagger})
\end{align}
implies that the function of self-adjoint operators (\ref{Weyl_main_definition}) is itself a self-adjoint operator. Moreover, one may demonstrate that  
\begin{align}\label{Weyl_derivative_of_function}
	f'_{\hat{A}_k} (\hat{A}_1, \ldots, \hat{A}_n) & \coloneqq \lim_{\epsilon\to 0} \frac 1{\epsilon} 
		\left[ f(\hat{A}_1, \ldots, \hat{A}_k + \epsilon, \ldots, \hat{A}_n)  - f(\hat{A}_1, \ldots, \hat{A}_k, \ldots, \hat{A}_n)\right]\nonumber\\
		&= \frac{1}{(2\pi)^n} \int \prod_{l=1}^n d\xi_l d\eta_l \, i\eta_k 
			\exp\left[ i \sum_{q=1}^n \eta_q(\hat{A}_q - \xi_q) \right] f(\xi_1, \ldots, \xi_n) \nonumber\\
		&= \frac{1}{(2\pi)^n} \int \prod_{l=1}^n d\xi_l d\eta_l
			f(\xi_1, \ldots, \xi_n) \left( -\frac{\partial}{\partial \xi_k} \right) \exp\left[ i \sum_{q=1}^n \eta_q(\hat{A}_q - \xi_q) \right]  \nonumber\\
		&= \frac{1}{(2\pi)^n} \int \prod_{l=1}^n d\xi_l d\eta_l
			\exp\left[ i \sum_{q=1}^n \eta_q(\hat{A}_q - \xi_q) \right] f'_{\xi_k} (\xi_1, \ldots, \xi_n).
\end{align}
Equation (\ref{Weyl_main_definition}) defines a one-to-one mapping between a function $f(\xi_1, \ldots, \xi_n)$ and a linear operator $f(\hat{A}_1, \ldots, \hat{A}_n)$. By the same token, Eq. (\ref{Weyl_derivative_of_function}) establishes a one-to-one mapping between the derivative of a function and the derivative of a linear operator.

The following theorem is of fundamental importance:
\begin{theorem}\label{Th_Weyl_commutator_theorem}
	Let $\hat{A}_1, \ldots, \hat{A}_n$ be some operators and $\hat{C}_k = \commut{\hat{A}_k, \hat{B}}$, $k=1,\ldots,n$. If 
	$\commut{ \hat{A}_k, \hat{C}_l } = \commut{  \hat{B}, \hat{C}_k } = 0$, $k,l=1,\ldots,n$, then
	\begin{align}\label{Weyl_MainTh_equality}
		\commut{ f(\hat{A}_1,\ldots, \hat{A}_n), \hat{B} } = \sum_{k=1}^n \commut{\hat{A}_k, \hat{B}} f_{\hat{A}_k}' (\hat{A}_1,\ldots, \hat{A}_n),	
	\end{align}
	where $f(\hat{A}_1,\ldots, \hat{A}_n)$ is defined by means of Eq. (\ref{Weyl_main_definition}).
\end{theorem}
\begin{proof}
	We introduce $\hat{A} \coloneqq i \sum_{q=1}^n \eta_q(\hat{A}_q - \xi_q)$ and $\hat{C} \coloneqq \commut{\hat{A}, \hat{B}} = i\sum_{q=1}^n \eta_q \hat{C}_q$;
	hence, $\commut{ \hat{A}, \hat{C} } = \commut{ \hat{B}, \hat{C} } = 0$. From the following identity:
	\begin{align}
		\commut{ \hat{A}_1 \cdots \hat{A}_n, \hat{B} }  &= 
			\sum_{k=1}^n \hat{A}_1 \cdots \hat{A}_{k-1} \commut{ \hat{A}_k, \hat{B} } \hat{A}_{k+1} \cdots \hat{A}_n,   
	\end{align}
	we obtain
	$
		\commut{ \hat{A}^k, \hat{B} } = k \hat{C} \hat{A}^{k-1}. 
	$
	It follows from Eq. (\ref{OperatorExp}) that
	\begin{align}
		\commut{ \exp(\hat{A}), \hat{B} } = \hat{C} \exp(\hat{A}) = \sum_{q=1}^n \hat{C}_q \frac{\partial}{\partial \hat{A}_q} \exp(\hat{A}).
	\end{align}
	Having substituted this equality into Eq. (\ref{Weyl_main_definition}), we finally reach Eq. (\ref{Weyl_MainTh_equality}).
\end{proof}
In the context of Maslov calculus \cite{Maslov1976, Nazaikinskii1992, Nazaikinskii1996}, theorem \ref{Th_Weyl_commutator_theorem} has been extended to a more general case where $\hat{C}_k$ need not commute with $\hat{A}_n$ and $\hat{B}$. In Ref. \cite{Transtrum2005},  commutators of the type $\commut{ f(\hat{A}_1,\ldots, \hat{A}_n), g(\hat{B}_1,\ldots, \hat{B}_n) }$ were considered.

\section{Unification of Classical and Quantum Mechanics}\label{Sec_Qunatum_to_Calssical}

We reiterate that the key difference between classical and quantum mechanics is that the operators of momentum and coordinate commute in the former case and do not commute in the latter case \cite{Dirac1958, Shirokov1976b, Shirokov1979d}. The operators $\hat{x}$, $\hat{p}$, $\hat{\lambda}_x$, and $\hat{\lambda}_p$ obeying the commutation relations 
\begin{align}\label{Complete_classical_algebra}
	\commut{ \hat{x}, \hat{\lambda}_x } = \commut{ \hat{p}, \hat{\lambda}_p } = i, \qquad
	\commut{ \hat{x}, \hat{p} } = \commut{ \hat{x}, \hat{\lambda}_p } = \commut{ \hat{p}, \hat{\lambda}_x } 
	= \commut{ \hat{\lambda}_x, \hat{\lambda}_p } = 0,
\end{align}
form {\it the classical algebra of operators}. Let us introduce another auxiliary algebra: The operators  $\hat{x}_q$, $\hat{p}_q$, $\hat{\vartheta}_x$, and $\hat{\vartheta}_p$ form {\it the quantum algebra of operators} satisfying
\begin{align}
	\commut{ \hat{x}_q, \hat{p}_q } = i\hbar\kappa, \quad (0 \leqslant \kappa \leqslant 1) \qquad
	\commut{ \hat{x}_q, \hat{\vartheta}_x } = \commut{ \hat{p}_q, \hat{\vartheta}_p } = i,
\end{align}
with all the other commutators vanishing. The operators $\hat{\vartheta}_x$ and $\hat{\vartheta}_p$ are introduced into the quantum algebra so that it resembles the classical algebra.  

The values of $\kappa$ in the domain $0 \leqslant \kappa \leqslant 1$ defines the quantum-to-classical character because the quantum algebra smoothly transforms into the classical one as $\kappa \to 0$. Since $\hbar$ enters in the canonical commutator relationship for the quantum coordinate and momentum as well as the time derivative in Sch\"{o}dinger equation, the limit $\hbar\to 0$ encompasses more than the criterion that the coordinate and momentum operators must commute in the classical limit. This situation motivated the introduction of the additional parameter $\kappa$.

To understand the transition from quantum to classical mechanics (see Fig. \ref{Fig_EhrenfestQuantization} for the roadmap of the current section), we first apply ODM to 
\begin{align}\label{EhrenfestThrs_in_quantumclass_case}
	m\frac{d}{dt} \bra{\Psi(t)} \hat{x}_q \ket{\Psi(t)} = \bra{\Psi(t)} \hat{p}_q \ket{\Psi(t)}, \qquad
	\frac{d}{dt} \bra{\Psi(t)} \hat{p}_q \ket{\Psi(t)} = \bra{\Psi(t)} -U'(\hat{x}_q) \ket{\Psi(t)}
\end{align}
and find the Hamiltonian $\hat{\mathcal{H}} = \mathcal{H}(\hat{x}_q, \hat{p}_q, \hat{\vartheta}_x, \hat{\vartheta}_p)$ such that $i\hbar \ket{d \Psi(t)/dt} = \hat{\mathcal{H}} \ket{\Psi(t)}$. Using theorem \ref{Th_Weyl_commutator_theorem}, we derive the system of  partial differential equations for the function $\mathcal{H}$,
\begin{align}
	\kappa m \frac{\partial \mathcal{H}}{\partial p_q} + \frac{m}{\hbar} \frac{\partial \mathcal{H}}{\partial \vartheta_x} = p_q, \qquad
	\kappa \frac{\partial \mathcal{H}}{\partial x_q} - \frac{1}{\hbar} \frac{\partial \mathcal{H}}{\partial \vartheta_p} = U'(x_q),
\end{align}
whose general solution reads
\begin{align}\label{MixedQuantumClassical_Hamiltonian}
	\hat{\mathcal{H}} = \frac{1}{\kappa} \left[ \frac{ \hat{p}_q^2 }{2 m} + U(\hat{x}_q) \right] 
					+ F \left( \hat{p}_q - \hbar\kappa \hat{\vartheta}_x, \hat{x}_q + \hbar\kappa \hat{\vartheta}_p \right),
\end{align}
where $F$ is an arbitrary differentiable function of two variables. 

Expression (\ref{MixedQuantumClassical_Hamiltonian}) was inferred only from the Ehrenfest theorems (\ref{EhrenfestThrs_in_quantumclass_case}) for the coordinate and momentum. We seek to show that $F$ remains free even if the Ehrenfest theorem is known for another observable $\hat{O} = O\left( \hat{x}_q, \hat{p}_q \right)$. The Ehrenfest theorem for $\hat{O}$ reads
\begin{align}
	i\hbar \frac{d}{dt} \bra{\Psi(t)} \hat{O}  \ket{\Psi(t)} 
	= \bra{\Psi(t)} \commut{\hat{O}, \hat{\mathcal{H}}} \ket{\Psi(t)}.
\end{align}
The contribution from $F$ to this Ehrenfest theorem would be measurable if $\hat{O}$ did not commute with  $\hat{F}$.
However, 
\begin{align}
	\commut{ \hat{p}_q, \hat{p}_q - \hbar\kappa \hat{\vartheta}_x } = \commut{ \hat{p}_q, \hat{x}_q + \hbar\kappa \hat{\vartheta}_p }
	= \commut{ \hat{x}_q, \hat{p}_q - \hbar\kappa \hat{\vartheta}_x } = \commut{ \hat{x}_q, \hat{x}_q + \hbar\kappa \hat{\vartheta}_p } 
	= 0 \Longrightarrow \commut{ \hat{O}, \hat{F} } = 0; \notag
\end{align}
thus, the function $F$ is truly undetectable.

We first set $F$ to zero and consider the Hamiltonian
\begin{align}\label{Hq_Hamlitonian}
	\hat{H}_q = \frac{1}{\kappa} \left[ \frac{ \hat{p}_q^2 }{2 m} + U(\hat{x}_q) \right],
\end{align}
to deduce whether the classical Liouville equation 
\begin{align}
	i \ket{d\Psi(t)/dt} = \hat{L} \ket{\Psi(t)}, \qquad
	\hat{L} = \hat{p} \hat{\lambda}_x / m - U'(\hat{x}) \hat{\lambda}_p + f(\hat{x}, \hat{p}), \qquad (\mbox{$f$ is an arbitrary function})
	\label{General_form_of_Liouvillian}
\end{align}
can be recovered from the Schr\"{o}dinger equation 
\begin{align}\label{SchEq_transition_from_quant_to_class}
	i\hbar \ket{d \Psi(t)/dt} = \hat{H}_q \ket{\Psi(t)}
\end{align}
as $\kappa\to 0$.

The classical and quantum algebras are the same, i.e., they are isomorphic. Indeed, the quantum operators ($\hat{x}_q$, $\hat{p}_q$, $\hat{\vartheta}_x$, and $\hat{\vartheta}_p$) can be constructed as linear combinations of the classical operators ($\hat{x}$, $\hat{p}$, $\hat{\lambda}_x$, and $\hat{\lambda}_p$) in infinitely many ways. Three examples of such realizations are: i) $\hat{x}_q = \hat{x}$, $\hat{p}_q = \hbar\kappa \hat{\lambda}_x + (\hbar\kappa -1)\hat{\lambda}_p$, $\hat{\vartheta}_x = \hat{\lambda}_x + \hat{\lambda}_p$, $\hat{\vartheta}_p = \hat{p} - \hat{x}$; ii) $\hat{x}_q = \hat{x} -\hbar\kappa \hat{\lambda}_p$, $\hat{p}_q = \hat{p}$, $\hat{\vartheta}_x = \hat{\lambda}_x$, $\hat{\vartheta}_p = \hat{\lambda}_p$; 
\begin{align}\label{ClassicQuantumAlg_isomorphism}
	\mbox{iii)} \quad \hat{x}_q = \hat{x} - \hbar\kappa \hat{\lambda}_p / 2, \quad \hat{p}_q = \hat{p} + \hbar\kappa \hat{\lambda}_x / 2, 
	\quad \hat{\vartheta}_x = \hat{\lambda}_x, \quad \hat{\vartheta}_p = \hat{\lambda}_p.
\end{align}
The linear isomorphism between the two algebras stimulates the question: Can $\hat{x}_q$ and $\hat{p}_q$ be expressed as linear combinations of the classical operators such that $\lim_{\kappa \to 0} \hat{x}_q = \hat{x}$, $\lim_{\kappa \to 0} \hat{p}_q = \hat{p}$, and $\lim_{\kappa \to 0} \hat{H}_q = \hbar \hat{L}$? Theorem \ref{QuantToClass_nogo_theorem_strong_version} negatively answers this question. First, we shall prove a more general statement:
\begin{lemma}\label{QuantToClass_nogo_lemma}
	Assume $U'(x) \not\equiv 0$. If there exist operators $\hat{x}_q$ and $\hat{p}_q$ such that they are linear combinations of $\hat{x}$, $\hat{p}$, $\hat{\lambda}_x$,  $\hat{\lambda}_p$, and 
	\begin{align}
		& \lim_{\kappa \to 0} \hat{x}_q = \hat{x}, \label{lemma_for_nogotheorem_eq1} \\ 
		& \lim_{\kappa \to 0} \hat{p}_q = \alpha \hat{p}, \label{lemma_for_nogotheorem_eq2} \\
		& \lim_{\kappa \to 0} \hat{H}_q = \beta\hbar \hat{L} \label{lemma_for_nogotheorem_eq3},
	\end{align}
	then $\commut{ \hat{x}_q, \hat{p}_q } = i \beta (\alpha + 1/\alpha) \hbar\kappa + o(\kappa)$.
\end{lemma}
\begin{proof}
	The asymptotic symbols $o(\kappa)$ and $o(1)$ are defined with respect to the limit $\kappa \to 0$.
	 By the condition of  the lemma, we set 
	\begin{align}\label{LinearExp_xq_pq}
		\hat{x}_q = \frac{\partial x_q}{\partial x} \hat{x} + \frac{\partial x_q}{\partial \lambda_x} \hat{\lambda}_x 
				+ \frac{\partial x_q}{\partial p} \hat{p} + \frac{\partial x_q}{\partial \lambda_p} \hat{\lambda}_p, \qquad
		\hat{p}_q = \frac{\partial p_q}{\partial x} \hat{x} + \frac{\partial p_q}{\partial \lambda_x} \hat{\lambda}_x 
				+ \frac{\partial p_q}{\partial p} \hat{p} + \frac{\partial p_q}{\partial \lambda_p} \hat{\lambda}_p.
	\end{align}
	From Eqs. (\ref{lemma_for_nogotheorem_eq1})-(\ref{lemma_for_nogotheorem_eq3}), 
	\begin{align}
		\beta\hbar \frac{\partial \hat{L}}{\partial \hat{\lambda}_x} = \lim_{\kappa \to 0} \frac{\partial \hat{H}}{\partial \hat{\lambda}_x}
		\Longrightarrow
		\beta\hbar \frac{\hat{p}}{m} = \lim_{\kappa \to 0} \left( \frac{ \hat{p}_q }{\kappa m} \frac{\partial p_q}{\partial \lambda_x}
				+ \frac{1}{\kappa} U'(\hat{x}_q) \frac{\partial x_q}{\partial \lambda_x} \right) 
			= \lim_{\kappa \to 0} \left( \frac{ \alpha \hat{p}}{\kappa m} \frac{\partial p_q}{\partial \lambda_x} 
				+ \frac{1}{\kappa} U'(\hat{x}) \frac{\partial x_q}{\partial \lambda_x} \right);
	\end{align}
	whence, we conclude that
	\begin{align}\label{pq_lambdax_and_xq_lambdax}
		\frac{\partial p_q}{\partial \lambda_x} =  \frac{\beta \hbar}{\alpha} \kappa + o(\kappa), \qquad 
		\frac{\partial x_q}{\partial \lambda_x} = o(\kappa).
	\end{align}
	By the same token, we derive from Eqs. (\ref{lemma_for_nogotheorem_eq1})-(\ref{lemma_for_nogotheorem_eq3}) that
	\begin{align}
		& \beta\hbar \frac{\partial \hat{L}}{\partial \hat{\lambda}_p} = \lim_{\kappa \to 0} \frac{\partial \hat{H}}{\partial \hat{\lambda}_p}
		\Longrightarrow
		- \beta\hbar U'(\hat{x}) = \lim_{\kappa \to 0} \left( \frac{ \alpha \hat{p}}{\kappa m} \frac{\partial p_q}{\partial \lambda_p} 
				+ \frac{1}{\kappa} U'(\hat{x}) \frac{\partial x_q}{\partial \lambda_p} \right)
		\Longrightarrow \notag\\
		& \frac{\partial x_q}{\partial \lambda_p} = -\beta\hbar \kappa + o(\kappa), \quad
		 \frac{\partial p_q}{\partial \lambda_p} = o(\kappa).
	\end{align}
	Equations (\ref{lemma_for_nogotheorem_eq1}) and (\ref{lemma_for_nogotheorem_eq2}) imply 
	\begin{align}\label{xq_x_and_pq_p}
		\frac{\partial x_q}{\partial x} = 1 + o(1), \qquad \frac{\partial x_q}{\partial p} = o(1), \qquad
		\frac{\partial p_q}{\partial p} = \alpha + o(1), \qquad \frac{\partial p_q}{\partial x} = o(1).
	\end{align}
	Substituting Eqs. (\ref{pq_lambdax_and_xq_lambdax})-(\ref{xq_x_and_pq_p}) into Eq. (\ref{LinearExp_xq_pq}) and calculating the commutator between $\hat{x}_q$ and $\hat{p}_q$ by using Eq. (\ref{Complete_classical_algebra}), we finalize the lemma's proof.
\end{proof}

\begin{theorem}[The strong version of the no-go theorem]\label{QuantToClass_nogo_theorem_strong_version}
	Assume $U'(x) \not\equiv 0$. There are no operators  $\hat{x}_q$ and $\hat{p}_q$ such that they are linear combinations of $\hat{x}$, $\hat{p}$, $\hat{\lambda}_x$,  $\hat{\lambda}_p$, and $\lim_{\kappa \to 0} \hat{x}_q = \hat{x}$, $\lim_{\kappa \to 0} \hat{p}_q = \hat{p}$, $\lim_{\kappa \to 0} \hat{H}_q = \hbar \hat{L}$, $\commut{ \hat{x}_q, \hat{p}_q } = i\hbar \kappa$.
\end{theorem}
\begin{proof}
	If such operators exist, then according to lemma \ref{QuantToClass_nogo_lemma}, $\commut{ \hat{x}_q, \hat{p}_q } = 2 i\hbar \kappa + o(\kappa)$, which contradicts the statement of the theorem.
\end{proof}

Let us consider the dependence of the wave function on $\kappa$. Theorem \ref{QuantToClass_nogo_theorem_strong_version} implies the following weaker statement, which can also be demonstrated independently: 
\begin{theorem}[The weak version of the no-go theorem]\label{QuantToClass_nogo_theorem_weak_version}
	Assume $U'(x) \not\equiv 0$. There are no operators $\hat{x}_q$ and $\hat{p}_q$ such that they are linear combinations of $\hat{x}$, $\hat{p}$, $\hat{\lambda}_x$,  $\hat{\lambda}_p$, and 
	\begin{align}
		& \lim_{\kappa \to 0} \left( \hat{x}_q - \hat{x} \right) \ket{\Psi_{\kappa}(t)} = 0, \label{strong_nogo_th_eq1} \\
		& \lim_{\kappa \to 0} \left( \hat{p}_q - \hat{p} \right) \ket{\Psi_{\kappa}(t)} = 0, \label{strong_nogo_th_eq2} \\
		& \lim_{\kappa \to 0} \left( \hat{H}_q - \hbar\hat{L} \right) \ket{\Psi_{\kappa}(t)} = 0, \label{strong_nogo_th_eq3}
	\end{align}
	and $\commut{ \hat{x}_q, \hat{p}_q } = i\hbar \kappa$.
\end{theorem}

This theorem might seem counterintuitive at first sight. To see that the result is correct, we need to elucidate the physical meaning of assumptions (\ref{strong_nogo_th_eq1})-(\ref{strong_nogo_th_eq3}). The quasi-classical wave function in the coordinate representation is known to be of the form
\begin{align}
	\Phi_{\kappa}(x,t) = F (x,t) \exp[ i S(x,t)/(\hbar\kappa) ],
\end{align}
where $S(x,t)$ satisfies the Hamilton-Jacobi equation as $\kappa \to 0$. Consider the action of the momentum operator, $\hat{p} = -i\hbar\kappa \partial /\partial x$, on this wave function
\begin{align}\label{pquantum_to_pclass}
	\hat{p} \Phi_{\kappa}(x,t) = \frac{\partial S(x,t)}{\partial x} \Phi_{\kappa}(x,t) + O(\kappa) \Longrightarrow
	\lim_{\kappa \to 0} \left[ \hat{p} - \mathrsfs{P}(x,t) \right] \Phi_{\kappa}(x,t) = 0,
\end{align}
where $\mathrsfs{P}(x,t) \coloneqq \partial S(x,t) / \partial x$ denotes a classical particle's momentum. Equation (\ref{pquantum_to_pclass}) coincides with condition (\ref{strong_nogo_th_eq2}), which means that the quantum momentum goes over to the classical momentum in the classical limit. Condition (\ref{strong_nogo_th_eq1}) implies the same for the coordinate. Nevertheless, one readily demonstrates that
\begin{align}\label{Hquantum_to_Hclass}
	\lim_{\kappa \to 0} \left[ \kappa\hat{H}_q - \mathrsfs{H}(x,t) \right] \Phi_{\kappa}(x,t)   = 0,
\end{align}
where $\mathrsfs{H}(x,t) \coloneqq \mathrsfs{P}^2 (x,t) / (2m) + U(x)$ is the classical Hamiltonian.  Equation (\ref{Hquantum_to_Hclass}) contradicts condition (\ref{strong_nogo_th_eq3}); thus, the statement of theorem \ref{QuantToClass_nogo_theorem_strong_version} is intuitively correct because the quantum Hamiltonian does not approach the Liouvillian in the classical limit. 

Now we face the dilemma: Hamiltonian (\ref{MixedQuantumClassical_Hamiltonian}) has been introduced as a generator of motion valid in both the classical ($\kappa\to 0$) and quantum ($\kappa \to 1$) cases, and yet the classical limit appears to be inconsistent. Condition (\ref{strong_nogo_th_eq3}) merely seems to demand that such a generalized generator of motion should become the Liouvillian in the classical limit. But, it does not. In fact, condition (\ref{strong_nogo_th_eq3}) imposes this restriction on Hamiltonian (\ref{Hq_Hamlitonian}), which is a special case ($F \equiv 0$) of more general Hamiltonian (\ref{MixedQuantumClassical_Hamiltonian}). The equality $\lim_{\kappa \to 0} \hat{\mathcal{H}} = \hbar \hat{L}$ is achievable for certain functions $F$. Before presenting a specific example of this function, let us prove the following two statements:

\begin{theorem}\label{Specification_of_F_Th1}
	Assume $F(p,x) =  Q(p) + G(x)$ and $U'(x) \not\equiv 0$. If there exist operators $\hat{x}_q$, $\hat{p}_q$, $\hat{\theta}_x$, and $\hat{\theta}_p$ such that they are linear combinations of $\hat{x}$, $\hat{p}$, $\hat{\lambda}_x$,  $\hat{\lambda}_p$, and
	\begin{align}\label{Conditions_Specification_of_F_Th1}
		 \lim_{\kappa \to 0} \hat{x}_q = \hat{x}, \qquad
		 \lim_{\kappa \to 0} \hat{p}_q = \hat{p}, \qquad
		 \lim_{\kappa \to 0} \hat{\theta}_x = \hat{\lambda}_x, \qquad
		 \lim_{\kappa \to 0} \hat{\theta}_p = \hat{\lambda}_p, \qquad
		 \lim_{\kappa \to 0} \hat{\mathcal{H}} = \hbar \hat{L},
	\end{align}
	then
	\begin{align}\label{Specification_of_F_Th1_Eqresult}
		F(p, x) = - \frac{p^2}{2 m \kappa} - \frac{1}{\kappa} U(x) + O(1). \qquad (\kappa \to 0) 
	\end{align}
\end{theorem}
\begin{proof}
	The current proof is similar to the proof of lemma \ref{QuantToClass_nogo_lemma}. From Eq. (\ref{Conditions_Specification_of_F_Th1}), we have 
	\begin{align}
		& \hbar \frac{\partial \hat{L}}{\partial \hat{\lambda}_x} 
			= \lim_{\kappa \to 0} \frac{\partial \hat{\mathcal{H}}}{\partial \hat{\lambda}_x}
		\Longrightarrow \notag\\
		& \hbar\frac{\hat{p}}{m} = \lim_{\kappa \to 0}\left[ \frac{\hat{p}}{\kappa m} \frac{\partial p_q}{\partial \lambda_x} 
			+ \frac{1}{\kappa} U'(\hat{x}) \frac{\partial x_q}{\partial \lambda_x} 
			+ F_1' (\hat{p}, \hat{x}) \left( \frac{\partial p_q}{\partial \lambda_x} 
				- \hbar \kappa \frac{\partial\theta_x}{\partial \lambda_x} \right)
			+ F_2' (\hat{p}, \hat{x}) \left( \frac{\partial x_q}{\partial \lambda_x} 
				+ \hbar \kappa \frac{\partial\theta_p}{\partial \lambda_x} \right)
			\right], \label{Specification_of_F_Th1_Eq1} \\
		& \hbar \frac{\partial \hat{L}}{\partial \hat{\lambda}_p} 
			= \lim_{\kappa \to 0} \frac{\partial \hat{\mathcal{H}}}{\partial \hat{\lambda}_p}
		\Longrightarrow \notag\\
		& -\hbar U'(\hat{x}) = \lim_{\kappa \to 0} \left[ \frac{\hat{p}}{\kappa m} \frac{\partial p_q}{\partial \lambda_p} 
			+ \frac{1}{\kappa} U'(\hat{x}) \frac{\partial x_q}{\partial \lambda_p} 
			+ F'_1(\hat{p}, \hat{x}) \left( \frac{\partial p_q}{\partial \lambda_p} 
				- \hbar\kappa \frac{\partial \theta_x}{\partial \lambda_p}\right) 
			+ F'_2(\hat{p}, \hat{x}) \left( \frac{\partial x_q}{\partial \lambda_p} 
				+ \hbar\kappa \frac{\partial \theta_p}{\partial \lambda_p} \right) 
			\right], \label{Specification_of_F_Th1_Eq2}
	\end{align}
	where $F'_1$ and $F'_2$ denote the partial derivatives of the function $F$ with respect to the first and second arguments, respectively [see Eq. (\ref{LinearExp_xq_pq}) regarding other notations].  Due to the assumption $F(p,x) = Q(p) + G(x)$, the expressions under the limits in Eqs. (\ref{Specification_of_F_Th1_Eq1}) and (\ref{Specification_of_F_Th1_Eq2}) can be represented as the sum of the term depending on $\hat{x}$ and the term depending on $\hat{p}$. For Eqs. (\ref{Specification_of_F_Th1_Eq1}) and (\ref{Specification_of_F_Th1_Eq2}) to be consistent, the first term must be of $o(1)$ in the case of Eq. (\ref{Specification_of_F_Th1_Eq1}), and the second term must be of $o(1)$ in Eq. (\ref{Specification_of_F_Th1_Eq2}). Since ${\partial\theta_x}/{\partial \lambda_x} = 1+o(1)$, ${\partial\theta_p}/{\partial \lambda_p} = 1 + o(1)$, $\partial p_q/\partial \lambda_x = o(1)$, and $\partial x_q / \partial \lambda_p = o(1)$, we obtain
	\begin{align}
		&  \hbar\frac{\hat{p}}{m} = \lim_{\kappa \to 0}\left[ \frac{\hat{p}}{\kappa m} \frac{\partial p_q}{\partial \lambda_x}  
			+ Q' (\hat{p}) \left( \frac{\partial p_q}{\partial \lambda_x} - \hbar \kappa \right) \right] 
			\Longrightarrow  Q' (\hat{p})  = - \frac{\hat{p}}{m\kappa} + O(1), \\
		& -\hbar U'(\hat{x}) = \lim_{\kappa \to 0} \left[ \frac{1}{\kappa} U'(\hat{x}) \frac{\partial x_q}{\partial \lambda_p} 
			+ G'(\hat{x}) \left( \frac{\partial x_q}{\partial \lambda_p} + \hbar\kappa \right) \right]
			\Longrightarrow G'(\hat{x}) = -\frac{1}{\kappa} U'(\hat{x}) + O(1).
	\end{align}
	It can be verified that the derived expressions for $Q' (\hat{p})$ and $G'(\hat{x})$ indeed simultaneously satisfy Eqs. (\ref{Specification_of_F_Th1_Eq1}) and (\ref{Specification_of_F_Th1_Eq2}). Hence, we finally reach Eq. (\ref{Specification_of_F_Th1_Eqresult}).
\end{proof}

\begin{theorem}\label{Specification_of_F_Th2}
	Assume $U'(x) \not\equiv 0$. If there exist operators $\hat{x}_q$, $\hat{p}_q$, $\hat{\theta}_x$, and $\hat{\theta}_p$ such that they are linear combinations of $\hat{x}$, $\hat{p}$, $\hat{\lambda}_x$,  $\hat{\lambda}_p$, and
	\begin{align}
		 \lim_{\kappa \to 0} \hat{x}_q = \hat{x}, \quad
		 \lim_{\kappa \to 0} \hat{p}_q = \hat{p}, \quad
		 \lim_{\kappa \to 0} \hat{\theta}_x = \hat{\lambda}_x, \quad
		 \lim_{\kappa \to 0} \hat{\theta}_p = \hat{\lambda}_p, \quad
		 \lim_{\kappa \to 0} \hat{\mathcal{H}} = \hbar \hat{L}, \quad
		 \frac{\partial x_q}{\partial \lambda_x} = o(\kappa), \quad
		 \frac{\partial p_q}{\partial \lambda_p} = o(\kappa), 
	\end{align}
	then
	\begin{align}\label{Specification_of_F_Th2_Eqresult}
		F(p, x) = - \frac{p^2}{2 m \kappa} - \frac{1}{\kappa} U(x) + O(1). \qquad (\kappa \to 0) 
	\end{align}
\end{theorem}
\begin{proof}
	Equation (\ref{Specification_of_F_Th2_Eqresult}) readily follows from Eqs. (\ref{Specification_of_F_Th1_Eq1}) and (\ref{Specification_of_F_Th1_Eq2}) after taking into account the following estimates: ${\partial x_q}/{\partial \lambda_x} = o(\kappa)$, $\partial x_q / \partial \lambda_p = o(1)$, ${\partial\theta_x}/{\partial \lambda_x} = 1+o(1)$, $\partial \theta_p / \partial \lambda_x = o(1)$, ${\partial p_q}/{\partial \lambda_p} = o(\kappa)$, $\partial p_q / \partial \lambda_x = o(1)$, $\partial \theta_x / \partial \lambda_p = o(1)$, and ${\partial\theta_p}/{\partial \lambda_p} = 1 + o(1)$.
\end{proof}

Theorems \ref{Specification_of_F_Th1} and \ref{Specification_of_F_Th2} quite explicitly specify permissible forms of the function $F$. Thus, the leading order term in Eq. (\ref{Specification_of_F_Th2_Eqresult}) shall be taken as the definition of the function $F$. Consider the Hamitonian
\begin{align}\label{Hamiltonian_Hqc}
	\hat{\mathcal{H}}_{qc} &\coloneqq \frac{1}{\kappa} \left[ \frac{ \hat{p}_q^2 }{2 m} + U(\hat{x}_q) \right] 
						- \frac{1}{2m\kappa} \left( \hat{p}_q - \hbar\kappa \hat{\vartheta}_x \right)^2 
						- \frac{1}{\kappa} U \left( \hat{x}_q + \hbar\kappa \hat{\vartheta}_p \right) \notag\\
					& \equiv \frac{\hbar}{m} \left( \hat{p}_q - \frac{\hbar\kappa}2 \hat{\theta}_x \right) \hat{\theta}_x
						+ \frac{1}{\kappa} \left[ U(\hat{x}_q) - U(\hat{x}_q + \hbar\kappa \hat{\theta}_p) \right].
\end{align} 
Comparing $\hat{\mathcal{H}}_{qc}$ [Eq. (\ref{Hamiltonian_Hqc})] with $\hat{L}$ [Eq. (\ref{General_form_of_Liouvillian})], we deduce that 
\begin{align}\label{Ansatse_ClassQuant_Alg_Isomorphism}
	\hat{p}_q - \hbar\kappa \hat{\theta}_x /2 = \hat{p}, \qquad \hat{\theta}_x = \hat{\lambda}_x, \qquad 
	\alpha \hat{x}_q + \beta \hbar\kappa \hat{\theta}_p = \hat{x},
\end{align}
where $\alpha$ and $\beta$ are unknown constants. Substituting Eqs. (\ref{Ansatse_ClassQuant_Alg_Isomorphism}) into Eq. (\ref{Hamiltonian_Hqc}) and requiring
\begin{align}\label{Quantum_to_Classical_limit}
	\lim_{\kappa \to 0} \hat{\mathcal{H}}_{qc} = \hbar \hat{L},
\end{align}
we conclude that only isomorphism (\ref{ClassicQuantumAlg_isomorphism}) between the classical and quantum algebras is permitted. The quantum variables $\hat{x}_q$ and $\hat{p}_q$ are known as the Bopp operators \cite{Bopp1956}. Hamiltonian (\ref{Hamiltonian_Hqc}) expressed solely in terms of the classical operators reads 
\begin{align}
	\hat{\mathcal{H}}_{qc} &= \frac{\hbar}{m} \hat{p} \hat{\lambda}_x + 
						\frac{1}{\kappa} U\left( \hat{x} - \frac{\hbar\kappa}{2} \hat{\lambda}_p \right) -
						\frac{1}{\kappa} U\left( \hat{x} + \frac{\hbar\kappa}{2} \hat{\lambda}_p \right).
\end{align}
Hamiltonian (\ref{Hamiltonian_Hqc}) exactly coincides with Liouvillian (\ref{General_form_of_Liouvillian}) for some potentials $U$. Let us find all such cases:
\begin{theorem}\label{Golden_Harmonic_Osc_Theorem}
	$\hat{\mathcal{H}}_{qc} \equiv \hbar \hat{L}$ if and only if the potential $U$ is a quadratic polynomial. 
\end{theorem}
\begin{proof}
	The $x\lambda_p$-representation of the classical algebra is
	\begin{align}\label{xlambda_p_representation}
		\hat{x} = x, \qquad \hat{\lambda}_x = -i \frac{\partial}{\partial x}, \qquad 
		\hat{p} = i \frac{\partial}{\partial \lambda_p}, \qquad \hat{\lambda}_p = \lambda_p.
	\end{align}
	Equation $\hat{\mathcal{H}}_{qc} = \hbar \hat{L}$ written in the $x\lambda_p$-representation leads to
	\begin{align}
		U( x - \alpha ) - U( x + \alpha) = -2\alpha U'(x),
		\qquad \alpha \coloneqq \hbar\kappa\lambda_p /2. 
	\end{align} 
	Fourier transforming this equation with respect to $x$, we obtain 
	\begin{align}\label{U_forier_transf_eq}
		\left[ \alpha \omega - \sin\left( \alpha \omega \right) \right] \tilde{U}(\omega) = 0,
	\end{align}
	where $\tilde{U}(\omega) = \int dx \, e^{-i\omega x} U(x) / \sqrt{2\pi}$. Since the equation $y = \sin y$ has the unique solution $y=0$, the non-trivial solution of Eq. (\ref{U_forier_transf_eq}) must be a distribution with support at the origin, whose most general form reads \cite{Gelfand1964}
	\begin{align}\label{Ufourer_deltafunc_expantion}
		\tilde{U} (\omega) = \sum_{n=0}^{\infty} c_n \delta^{(n)} (\omega).
	\end{align}
	From the identity (see, e.g., Ref. \cite{Kanwal1998})
	\begin{align}
		x^m \delta^{(n)} (x) = \left\{
		\begin{array}{ccc}
			(-1)^m \frac{n!}{(n-m)!} \delta^{(n-m)}(x) & \mbox{ if } & m < n, \\
			(-1)^n n! \delta (x)				& \mbox{ if } & m=n, \\
			0							& \mbox{ if } & m > n,
		\end{array} \right. \notag
	\end{align}
	one derives 
	\begin{align}
		& \left[ \alpha \omega - \sin(\alpha \omega) \right] \delta^{(2n+1)} (\omega) = \sum_{m=1}^{n} (-1)^m \alpha^{2m+1} 
			\left( 2n+1 \atop 2m+1 \right) \delta^{(2n-2m)}(\omega), \notag \\
		& \left[ \alpha \omega - \sin(\alpha \omega) \right] \delta^{(2n)} (\omega) = \sum_{m=1}^{n-1} (-1)^m \alpha^{2m+1}
			\left( 2n \atop 2m+1 \right) \delta^{(2n-2m-1)} (\omega). \notag
	\end{align}
	From these equations It follows that the first three terms in expansion (\ref{Ufourer_deltafunc_expantion}) are arbitrary. Moreover, substituting expansion (\ref{Ufourer_deltafunc_expantion}) into Eq. (\ref{U_forier_transf_eq}), we obtain
	\begin{align}
		& \sum_{n=0}^{\infty} c_{n} \left[ \alpha \omega - \sin\left( \alpha \omega \right) \right] 
			\delta^{(n)} (\omega) = \sum_{n=0}^{\infty} f_{2n+1} \left( \alpha \right) \delta^{(2n)}(\omega)
			+ \sum_{n=1}^{\infty} f_{2n} \left( \alpha \right) \delta^{(2n-1)}(\omega) = 0, \label{U_forier_transf_eq_expanded}
	\end{align}
	where
	\begin{align}
		f_p (\alpha) = \sum_{m=1}^{\infty} (-1)^m \alpha^{2m+1} c_{p+2m} \left( p + 2m \atop 2m+1 \right). \notag
	\end{align}
	Equation (\ref{U_forier_transf_eq_expanded}) must be satisfied for all $\alpha$, then $f_p(\alpha) \equiv 0$, $\forall p$. This condition implies that $c_n = 0$, $n=3,4,5,\ldots$ because the functions $f_p$ are analytic by construction. Therefore, the most general solution of Eq. (\ref{U_forier_transf_eq}) reads
	\begin{align}
		\tilde{U}(\omega) = c_0 \delta(\omega) + c_1 \delta'(\omega) + c_2 \delta'' (\omega). \notag
	\end{align}
	Finally, we note that the inverse Fourier transform of this function is a quadratic polynomial. 
\end{proof}

\section{Derivation of the Canonical Quantization Rule}\label{Sec_CanonicalQuantz}

The equation of motion for a classical observable $f=f(p,x)$ reads
\begin{align}
	df/dt = \Poisson{f, \mathrsfs{H}},
\end{align}
where $\mathrsfs{H} = \mathrsfs{H}(p,x)$ is the classical Hamiltonian and $\poisson{\cdot, \cdot}$ denotes the Poisson bracket
\begin{align}
	\Poisson{f, g} \coloneqq \frac{\partial f}{\partial x} \frac{\partial g}{\partial p} - \frac{\partial f}{\partial p} \frac{\partial g}{\partial x}.
\end{align}

We apply ODM to the following equation of motion
\begin{align}\label{canonical_quantiz_EhrenfestTh}
	\frac{d}{dt} \bra{\Psi(t)} \hat{f} \ket{\Psi(t)} = \bra{\Psi(t)} \widehat{\Poisson{f, \mathrsfs{H}}} \ket{\Psi(t)},
\end{align}
where $\widehat{\Poisson{f, \mathrsfs{H}}}$ denotes a quantum analog of the Poisson bracket, which is to be found.

From Stone's theorem (see Sec. \ref{Sec_Stones_Th}), we conclude that there exits a unique self-adjoint operator $\hat{H}$ such that 
$i\hbar \ket{d\Psi(t)/dt} = \hat{H} \ket{\Psi(t)}$. Therefore, we obtain from Eq. (\ref{canonical_quantiz_EhrenfestTh})
\begin{align}
	\bra{\Psi(t)} \commut{\hat{f}, \hat{H}} \ket{\Psi(t)} = i\hbar \bra{\Psi(t)} \widehat{\Poisson{f, \mathrsfs{H}}} \ket{\Psi(t)}.
\end{align}
The stronger version of this equality gives the celebrated canonical quantization rule
\begin{align}
	\widehat{\Poisson{f, \mathrsfs{H}}} \equiv \frac{1}{i\hbar} \commut{\hat{f}, \hat{H}}.
\end{align}

\section{``Stripping'' the Schwinger Quantum Action Principle}\label{Sec_SchwingerPrinciple}

The Schwinger quantum action principle \cite{Schwinger1951, Schwinger1953, Schwinger2003, Toms2007} states that a propagator's variation is proportional to the matrix element of the quantum action's variation, 
\begin{align}
	\delta \langle \alpha, t \ket{\beta, 0} = (i/\hbar) \bra{\alpha, t} \delta\hat{S}(t) \ket{\beta, 0}.
\end{align}
Consider the case when the action's variation is induced by the system's dynamics, i.e.,
$\delta\hat{S}(t) = \frac{\partial \hat{S}(t)}{\partial t} \delta t.$

\begin{align}\label{SchwingerVarPric_Adapted}
	\frac{\delta}{\delta t} \langle \alpha, t \ket{\beta, 0}  \equiv \frac{\partial}{\partial t} \langle \alpha, t \ket{\beta, 0} 
	= (i/\hbar) \bra{\alpha, t} \frac{\partial\hat{S}(t)}{\partial t} \ket{\beta, 0}.
\end{align}
According to Stone's theorem (see Sec. \ref{Sec_Stones_Th}), 
\begin{align}
	i\hbar \frac{d}{dt} \hat{U}(t) = \hat{H} \hat{U}(t), \qquad \ket{\alpha, t} = \hat{U}(t) \ket{\psi_{\alpha}}. \nonumber
\end{align}
Equation (\ref{SchwingerVarPric_Adapted}) can be rewritten as
\begin{align}
	-i\hbar \frac{\partial}{\partial t} \bra{\psi_{\alpha}} \hat{U}^{\dagger}(t) \ket{\beta, 0} = 
	\bra{\psi_{\alpha}} \hat{U}^{\dagger}(t) \frac{\partial\hat{S}(t)}{\partial t} \ket{\beta, 0}, \nonumber\\
	\bra{\psi_{\alpha}} \hat{U}^{\dagger}(t)\hat{H} \ket{\beta, 0} =  \bra{\psi_{\alpha}} \hat{U}^{\dagger}(t) \frac{\partial\hat{S}(t)}{\partial t} \ket{\beta, 0}.
\end{align}
Since the previous equation is valid for any $\ket{\psi_{\alpha}}$ and $\ket{\beta, 0}$, we conclude that
\begin{align}
	\hat{H} =  \frac{\partial\hat{S}(t)}{\partial t},
\end{align}
which is an operator analogue of the Hamilton-Jacobi equation.

Therefore, the Hamilton-Jacobi equation lies behind the Schwinger quantum action principle. Note, however, that the Schwinger principle is more general than the Hamilton-Jacobi equation because it holds for any type of variations (for further details see Refs. \cite{Schwinger1951, Schwinger1953, Schwinger2003, Toms2007}).

\section{Measuring Time in Quantum Mechanics}\label{Sec_TimeMeasuring_QM}

Time has a peculiar role in quantum mechanics. First and foremost, it is a key independent variable in dynamical equations. However, the definition of the self-adjoint operator representing the time observable is quite challenging, and it is still a topic of on-going discussions (see, e.g., reviews \cite{Muga2000, deCarvalho2002, Maji2007}). In this section we attempt to address the problem of defining the time observable in quantum mechanics from the point of view of ODM. Note that all the concepts put forward in this section are also applicable to classical mechanics once the Hamiltonian $\hat{H}$ is substituted by the Liouvillian $\hat{L}$.

Time is physically defined and measured by clocks. We come to know the current time by simply observing the position of a clock's pointer. In other words, we obtain the value of time indirectly -- by measuring some other observable. Assume there is an observable, represented by a self-adjoint operator $\hat{\mathcal{T}}$, such that its average evolves as 
\begin{align}\label{Atemp1_def_evolution_of_T}
	\bra{\Psi(t)} \hat{\mathcal{T}} \ket{\Psi(t)}  = \alpha t,
\end{align}
where $\alpha$ is a real constant. Then, the equality $t \coloneqq \bra{\Psi} \hat{\mathcal{T}} \ket{\Psi} / \alpha$ can be taken as the definition of time through the observable $\hat{\mathcal{T}}$. The form of Eq. (\ref{Atemp1_def_evolution_of_T}) allows for the direct application of ODM to find the operator $\hat{\mathcal{T}}$,
$$
	\frac{d}{dt} \bra{\Psi(t)} \hat{\mathcal{T}} \ket{\Psi(t)}  = \alpha;
$$
whence, we obtain the commutator equation for the unknown operator
\begin{align}\label{commut_relation_for_HT}
	i \commut{ \hat{H}, \hat{\mathcal{T}} } = \alpha \hbar.
\end{align}
The classes of operators $\hat{H}$ and $\hat{\mathcal{T}}$ obeying the canonical commutation relation (\ref{commut_relation_for_HT}) are generally quite restrictive. Pauli \cite{Pauli1958} formally demonstrated that the existence of the self-adjoint time operator canonically conjugate to the Hamiltonian implies that both operators possess completely continuous spectra spanning the entire real line. This statement is known as the Pauli theorem. However, such a theorem does not withstand a thorough and rigorous analysis \cite{Galapon2002} and is incorrect in general. Recall that the canonical conjugation of the operators $\hat{\mathcal{T}}$ and $\hat{H}$ is also a theoretical justification for the energy-time uncertainty relation.

There are other ways to define time measurement in quantum mechanics via ODM. We present the following two possibilities: First, assuming that there exist a self-adjoint operator $\hat{\mathcal{T}}_1$ and a real valued function $f( \cdot )$ such that 
\begin{align}
	\bra{ \Psi(t) } \hat{\mathcal{T}}_1 \ket{ \Psi(t) } = \alpha t \bra{ \Psi(t) } f(\hat{H}) \ket{ \Psi(t) },
\end{align}
i.e., $\alpha t \coloneqq \bra{ \Psi } \hat{\mathcal{T}}_1 \ket{ \Psi } / \bra{ \Psi } f(\hat{H}) \ket{ \Psi }$, we readily obtain the commutator equation 
\begin{align}\label{commut_relation_for_HT1}
	i \commut{ \hat{H}, \hat{\mathcal{T}}_1 } = \alpha \hbar f(\hat{H}).
\end{align}
The second method allows for the observable to be time-dependent. Consider the equation
\begin{align}
	\bra{ \Psi(t) } \hat{\mathcal{T}}_2(t) \ket{ \Psi(t) } = \alpha (t),
\end{align}
assuming $\alpha(t)$ is an invertible function, such that the instantaneous value of time can be defined as
\begin{align}
	t \coloneqq \alpha^{-1} \left( \bra{ \Psi } \hat{\mathcal{T}}_2 \ket{ \Psi } \right).
\end{align}
Then, the equation for the unknown operator $\hat{\mathcal{T}}_2(t)$ reads
\begin{align}
	i \commut{ \hat{H}, \hat{\mathcal{T}}_2(t) } + \hbar \hat{\mathcal{T}}_2' (t) = \hbar \alpha'(t).
\end{align}

\section{Quantization in Curvilinear Coordinates}

\subsection{Quantum Mechanics in Curvilinear Coordinates}\label{CurvilinearCoordWaveFunc}

In this section we extend ODM to $n$-dimensional curved spaces as well as 
to Euclidean spaces in curvilinear coordinates, which are important special cases. 

The classical Hamiltonian of interest is a scalar invariant of the form 
\begin{align}
	H = \frac{1}{2m}  P_{\mu}  g^{\mu \nu} P_{\nu} + U(X).
\end{align}
Note that the Einstein summation convention is assumed throughout. The classical Hamilton equations of motion are
\begin{align}
	 \dot{X^\sigma} =& \frac{\partial H }{\partial P_{\sigma} }  =  \frac{1}{m} g^{\sigma \mu} P_{\mu},  \label{HE1}  \\
	 \dot{ P_\sigma} =& -\frac{\partial H }{\partial X^{\sigma} } =  
	 - \frac{1}{2m}  P_{\mu}  \frac{\partial g^{\mu \nu}}{\partial X^{\sigma}} P_{\nu} - \frac{\partial U }{\partial X^\sigma}.
	 \label{HE2}
\end{align}
Recall that a non-degenerate coordinate transformation 
\begin{align}\label{CoordTransformXTilde}
 X^{\mu} = X^{\mu}( \tilde{X} ) 
 \end{align}
 induces the corresponding momentum transformation 
 \begin{align}\label{MomentumTransformXTilde}
  \tilde{P}_{\mu} =  P_{\nu} \frac{\partial X^{\nu}}{\partial \tilde{X}^{\mu}  } 
  \Longrightarrow
   \frac{\partial \tilde{X}^{\alpha}}{\partial X^{\beta}} 
 	= \frac{\partial P_{\beta}}{\partial \tilde{P}_{\alpha}}. 
 \end{align}
Direct calculations show that the Poisson brackets calculated with respect to $(\tilde{X}, \tilde{P})$ obey
 \begin{align}
 \Poisson{ P_{\mu}, P_{\nu}  } =  \Poisson{ X^{\mu}, X^{\nu}  } = 0, \qquad&
 \Poisson{ X^{\mu}, P_{\nu}  } = \delta^{\mu}_{\,\,\,\nu};
 \end{align}
 therefore, the transformations (\ref{CoordTransformXTilde}) and (\ref{MomentumTransformXTilde}) are canonical.

According to ODM, the classical coordinates are replaced by expectation values of the corresponding 
self-adjoint quantum operators. However,  extra caution is required to consistently define the operators in curvilinear coordinates. The probability density 
is calculated as
\begin{align}
	 d\mathcal{P}  =  \psi^* \psi \sqrt{g}\, dX^1 dX^2 \cdots dX^n  =   \psi^* \psi \sqrt{g}\, d^n X 
\end{align}
with the essential presence of the weight $\sqrt{g}$, which can be 
absorbed as part of the wave function by defining the weighted wave function 
$\boldsymbol{\psi}$, 
\begin{align}
	\boldsymbol{\psi} = g^{1/4}\psi.
\end{align}
The scalar expectation value is calculated by means of the weighted integral
\begin{align}\label{WeightedAveragingDef}
	\langle  \hat{F} \rangle  = \int \psi^* F \psi \sqrt{g}\,d^nX, 
\end{align}
which can be expressed as
\begin{align}
	\langle  \hat{F} \rangle =
  	 \int \boldsymbol{\psi}^* g^{-\frac{1}{4}}\hat{F} \left( g^{-\frac{1}{4}} \boldsymbol{\psi} \right) \sqrt{g} \, d^n X 
 	=  \int \boldsymbol{\psi}^* g^{\frac{1}{4}}\hat{F} \left(g^{-\frac{1}{4}} \boldsymbol{\psi} \right)  d^n X.
\end{align}
According to this scheme, the scalar 
weighted operator $\boldsymbol{\hat{F}}$ and the weighted wave function $ \boldsymbol{\psi}$  must transform 
according to the following rules
\begin{align}
	  \boldsymbol{\psi} \rightarrow \boldsymbol{\overline{\psi}} = g^{1/4} \boldsymbol{\psi}, \qquad
	  \hat{\boldsymbol{F}} \rightarrow  \overline{\boldsymbol{\hat{F}}}  = g^{1/4} \hat{\boldsymbol{F}} g^{-1/4}, 
\end{align}
such that the expectation value of the weighted operator reads
\begin{align}
	  \langle \hat{\boldsymbol{F}} \rangle = \int \boldsymbol{\psi}^* \hat{\boldsymbol{F}} \boldsymbol{\psi} d^n X,
\end{align}
which is an invariant scalar under coordinate transformations.
With this in mind, the Ehrenfest theorem applied to the first Hamilton 
equation (\ref{HE1}) can be written as
\begin{align}
	\frac{d}{dt} \langle \hat{x}^{\sigma}  \rangle =
	\left\langle  
	 \frac{1}{2m} g^{\sigma \mu} g^{\frac{1}{4}}  \hat{p}_{\mu} g^{-\frac{1}{4}} + 
	  \frac{1}{2m}  g^{-\frac{1}{4}}  \hat{p}_{\mu} g^{\frac{1}{4}} g^{\sigma \mu} 
	 \right\rangle,
\end{align}
where the second term was added in the right hand side to enforce Hermiticity; whence,
\begin{align}
	\frac{d}{dt} \langle \hat{x}^{\sigma}  \rangle =
	\left \langle  
	\frac{\partial }{\partial \hat{p}_\sigma}\left(
	 \frac{1}{2m} g^{-\frac{1}{4}}  \hat{p}_{\nu} g^{\frac{1}{4}}   g^{\nu \mu} g^{\frac{1}{4}}  \hat{p}_{\mu} g^{-\frac{1}{4}}
	\right)
 	\right\rangle,
\end{align}
while the second Hamilton equation (\ref{HE2}) leads to
\begin{align}
	\frac{d}{dt} \langle \hat{p}_{\sigma}  \rangle  =  
	-\left \langle
	 \frac{\partial }{\partial \hat{x}^\sigma}\left(
	 \frac{1}{2m} g^{-\frac{1}{4}}  \hat{p}_{\nu} g^{\frac{1}{4}}   g^{\nu \mu} g^{\frac{1}{4}}  \hat{p}_{\mu} g^{-\frac{1}{4}}
	\right)
	- \frac{\partial U }{\partial \hat{x}^\sigma}
	\right \rangle.
\end{align}
According to Stone's theorem (see Sec. \ref{Sec_Stones_Th}),
\begin{align}
 	i \hbar \frac{\partial }{\partial t} \boldsymbol{\psi} = \boldsymbol{\mathcal{H}} \boldsymbol{\psi} ,
\end{align}
the derivatives of the expectation values over time are replaced by 
commutators,
\begin{align}
	\frac{d}{dt} \langle \hat{x}^{\sigma}  \rangle   = 
	 \frac{i}{\hbar} \left\langle [ \boldsymbol{\mathcal{H}} , \hat{x}^{\sigma}  ] \right\rangle, 
 	\qquad
	\frac{d}{dt} \langle \hat{p_{\sigma}}  \rangle =
	\frac{i}{\hbar} \left\langle [ \boldsymbol{\mathcal{H}} , \hat{p}_{\sigma}  ] \right\rangle. 
\end{align}
Lifting the averaging leads to the following equations for the unknown quantum Hamiltonian: 
\begin{align}
	\frac{i}{\hbar} \Commut{ \boldsymbol{\mathcal{H}} , \hat{x}^{\sigma}  } = \frac{\partial }{\partial \hat{p}_\sigma}\left(
	 \frac{1}{2m} g^{-\frac{1}{4}}  \hat{p}_{\nu} g^{\frac{1}{4}}   g^{\nu \mu} g^{\frac{1}{4}}  \hat{p}_{\mu} g^{-\frac{1}{4}}
	\right),  
	\qquad
	\frac{i}{\hbar} \Commut{ \boldsymbol{\mathcal{H}} , \hat{p}_{\sigma}  } = 
	 - \frac{\partial }{\partial \hat{x}^\sigma}\left(
	 \frac{1}{2m} g^{-\frac{1}{4}}  \hat{p}_{\nu} g^{\frac{1}{4}}   g^{\nu \mu} g^{\frac{1}{4}}  \hat{p}_{\mu} g^{-\frac{1}{4}}
	\right)
	- \frac{\partial U }{\partial x^\sigma}.
\end{align}
The latter pair of equations reduces to 
\begin{align}
	 \frac{\partial \boldsymbol{\mathcal{H}}}{\partial p_\sigma} = \frac{\partial }{\partial p_\sigma}\left(
	 \frac{1}{2m} g^{-\frac{1}{4}}  \hat{p}_{\nu} g^{\frac{1}{4}}   g^{\nu \mu} g^{\frac{1}{4}}  \hat{p}_{\mu} g^{-\frac{1}{4}}
	\right),  
	\qquad
	 -\frac{\partial \boldsymbol{\mathcal{H}}}{\partial x^\sigma} = 
	 - \frac{\partial }{\partial x^\sigma}\left(
	 \frac{1}{2m} g^{-\frac{1}{4}}  \hat{p}_{\nu} g^{\frac{1}{4}}   g^{\nu \mu} g^{\frac{1}{4}}  \hat{p}_{\mu} g^{-\frac{1}{4}}
	\right)
	- \frac{\partial U }{\partial x^\sigma},
\end{align}
after postulating the canonical commutation relations
\begin{align}
	\Commut{ \hat{x}^\sigma , \hat{p}_{\mu} } = \delta^{\sigma}_{\,\,\,\mu} i\hbar,
\end{align}
Therefore, the Hamiltonian reads
\begin{align}\label{PodolskyHamiltonian}
 	\boldsymbol{\mathcal{H}} = \frac{1}{2m} g^{-1/4} \hat{p}_{\mu} g^{1/4} g^{\mu \nu} g^{1/4}\hat{p}_{\nu} g^{-1/4} + U(\hat{x}).
\end{align}
The problem of constructing 
consistent quantum Hamiltonians directly in curvilinear coordinates without having to
rely on the Hamiltonian in Cartesian coordinates was solved by Podolsky \cite{Podolsky1928}. In particular he derived Hamiltonian 
(\ref{PodolskyHamiltonian}). 

Note that the approach presented above based on ODM is equivalent to standard tensor calculus methods. By definition, a scalar $\boldsymbol{\Phi}$ of weight $N$ 
transforms as
\begin{align}
	 \boldsymbol{\Phi} \rightarrow
 	  \overline{\boldsymbol{\Phi}} = \left |\frac{\partial x}{\partial \overline{x}} \right|^N 
	 {\boldsymbol{ \Phi}} = 
 	 g^{N/2} \boldsymbol{\Phi},
\end{align}
with the covariant derivative being
\begin{align}
	 \nabla_{\mu}  \boldsymbol{\Phi}  = 
 	 \frac{\partial  \boldsymbol{\Phi}  }{\partial x^{\mu}} - N \Gamma^{\alpha}_{\,\,\,\mu\alpha} \boldsymbol{\Phi}.
\end{align}
Similarly, a scalar $\boldsymbol{\Phi}$ of weight $1/2$ transforms as 
\begin{align}
 	\boldsymbol{\Phi} \rightarrow \boldsymbol{\overline{\Phi}} = g^{1/4} \boldsymbol{\Phi}   
\end{align}
such that the covariant derivative reads
\begin{align}
 	 \nabla_{\mu} \boldsymbol{\Phi} = 
	 \frac{\partial  \boldsymbol{\Phi}  }{\partial x^{\mu}} - \frac{1}{2} \Gamma^{\alpha}_{\,\,\,\mu\alpha} \boldsymbol{\Phi}
	   =   g^{1/4} \partial_{\mu}\,  g^{-1/4} \boldsymbol{\Phi}.
\end{align}
As a result, the following transformation rule takes place
\begin{align}
 	\nabla_{\mu} \boldsymbol{\Phi}   \rightarrow
 	  \overline{\nabla}_{\mu}  \overline{\boldsymbol{\Phi}}   
 	 = g^{1/4}  \frac{\partial x^{\alpha}}{\partial \overline{x}^{\mu}} 
 	  \nabla_{\mu} \boldsymbol{\Phi}. 
\end{align}
The contraction of a  contravariant tensor $V^{\nu}$ of weight $1/2$ with the covariant derivative is 
\begin{align}
 	 \nabla_{\mu} \boldsymbol{V}^{\mu} = 
	\frac{\partial \boldsymbol{V}^{\mu}}{\partial x^{\mu}} + \frac{1}{2} \boldsymbol{V}^{\nu} \Gamma^{\alpha}_{\,\,\,\mu\alpha}
 	=  g^{-1/4} \partial_{\mu} \left(  g^{1/4} \boldsymbol{V}^{\mu} \right).
\end{align}
This expression can be used to construct a scalar of weight $1/2$ 
by contraction,
\begin{align}
 	  \nabla_{\mu} g^{\mu\nu} \nabla_{\nu} \boldsymbol{\Phi} = 
  	  g^{-1/4} \partial_{\mu} \left(  g^{1/4} g^{\mu \nu}   g^{1/4} \partial_{\nu} g^{-1/4}  \boldsymbol{\Phi} \right),
\end{align}
such that the following transformation rule is obeyed
\begin{align}
	   \nabla_{\mu} g^{\mu\nu} \nabla_{\nu} \boldsymbol{\Phi} 
	 \rightarrow 
 	    \nabla_{\mu} \overline{g}^{\mu\nu} \nabla_{\nu} \overline{\boldsymbol{\Phi}} =
	     g^{1/4}  \nabla_{\mu} g^{\mu\nu} \nabla_{\nu} \boldsymbol{\Phi},  
\end{align}
which ultimately leads us to the conclusion that
$ 
 	 \int  \boldsymbol{\Phi}^*   \nabla_{\mu} g^{\mu\nu} \nabla_{\nu} \boldsymbol{\Phi} d^nx
$ 
is an invariant scalar under coordinate transformations.

In the context of weighted wave functions, the partial derivative is neither invariant under transformations nor a scalar.  This problem 
can be partially solved by defining the weighted momentum operator through the 
covariant derivative,
\begin{align}
	  \boldsymbol{\hat{p}_{\mu}} = -i \hbar \nabla_{\mu}.
\end{align}
The explicit form of the covariant derivative is problem dependent, e.g., if the wave function is a scalar of weight $1/2$, then the covariant derivative takes the form 
\begin{align}
	  \nabla_{\mu} \boldsymbol{\psi} =   g^{1/4} \partial_{\mu}\,  g^{-1/4} \boldsymbol{\psi}.
\end{align}
Even though such a momentum operator is properly defined, it may not necessarily be self-adjoint, and therefore, may not possess  physical expectation values. Moreover, it is more natural to calculate the expectation values of
the contravariant components, which are dimensional quantities with fixed physical meanings.

\subsection{Phase Space Representation in Curvilinear Coordinates}

Let us derive the phase space representation of quantum dynamics in curvilinear coordinates employing ODM.  Following the recipe in Sec. \ref{Sec_Qunatum_to_Calssical}, we begin by extending 
the quantum algebra with the auxiliary operators $\hat{\theta}^{\nu}$ and $\hat{\lambda}^{\nu}$ such that
\begin{align}
	  \commut{ \hat{x}^{\mu} , \hat{p}_{\nu} } = i \hbar \kappa \delta^{\mu}_{\,\,\,\,\nu}, \qquad
	  \commut{ \hat{x}^{\mu} , \hat{\lambda}_{\nu} } = i \delta^{\mu}_{\,\,\,\,\nu}, \qquad
	  \commut{ \hat{p}_{\mu} , \hat{\theta}^{\nu} } = i \delta^{\nu}_{\,\,\,\,\mu}, \qquad
	  \commut{ \hat{\lambda}_{\mu} , \hat{\theta}^{\nu} } = 0,  
	  \label{extended-quantum-algebra}
\end{align} 
where $\kappa$ is a measure of quantumness/commutativity.  

Note that contrary to Sec. \ref{CurvilinearCoordWaveFunc}, the averaging here does not require the weight $\sqrt{g}$,
\begin{align}
\langle  \hat{F} \rangle  = \int \psi^* F \psi \,d^nX d^nP
\end{align}
[compare this with Eq. (\ref{WeightedAveragingDef})]. According to Stone's theorem (Sec. \ref{Sec_Stones_Th}), the generator of dynamics $\boldsymbol{\mathcal{W}}$ in the phase space is introduced as
\begin{align}\label{W_curvil_cordinate_Stone}
	  i \hbar \frac{\partial }{\partial t} \boldsymbol{\psi} =   \boldsymbol{\mathcal{W}} \boldsymbol{\psi}.
\end{align}
The generator $\boldsymbol{\mathcal{W}}$ must satisfy the following system of equations
\begin{align}
	  \kappa \frac{\partial \boldsymbol{\mathcal{W}} }{\partial \hat{p}_\sigma} + \frac{1}{\hbar}\frac{\partial  \boldsymbol{\mathcal{W}}  }{\partial \hat{\lambda}_{\sigma} }  
	  &= \frac{\partial }{\partial \hat{p}_\sigma}\left( \frac{1}{2m} 
	     \hat{p}_{\nu}  g^{\nu \mu} \hat{p}_{\mu} \right),  \\
	  -\kappa \frac{\partial  \boldsymbol{\mathcal{W}} }{\partial \hat{x}^\sigma} + \frac{1}{\hbar}\frac{\partial  \boldsymbol{\mathcal{W}}  }{\partial \hat{\theta}^{\sigma} }
	  &= - \frac{\partial }{\partial \hat{x}^\sigma}\left( \frac{1}{2m} \hat{p}_{\nu}  g^{\nu \mu}  \hat{p}_{\mu}    \right) - \frac{\partial U }{\partial \hat{x}^\sigma}.
	  \label{W-equations}
\end{align}
The solutions of Eq. (\ref{W-equations}) is 
 \begin{align}
	  \boldsymbol{\mathcal{W}} = \frac{1}{2m\kappa} \hat{p}_{\mu} g^{\mu \nu} \hat{p}_{\nu} + \frac{1}{\kappa}U(\hat{x}) 
	      + f\left( \hat{x}^{\mu} + \hbar \kappa \hat{\theta}^{\mu} , \hat{p}_{\mu} - \hbar \kappa \hat{\lambda}_{\mu} \right), 
	  \label{Wsol1}
\end{align}
where $f$ is an arbitrary function. The quantum algebra (\ref{extended-quantum-algebra}) can be realized in terms of the classical operators $\hat{X}^{\mu}$, $\hat{P}_{\mu}$, 
$\hat{\Theta}^{\mu}$, and $\hat{\Lambda}_{\mu}$,
\begin{align}
	  \hat{x}^{\mu} = \hat{X}^{\mu} - \frac{\hbar\kappa}{2} \hat{\Theta}^{\mu}, \qquad
	  \hat{p}_{\mu} = \hat{P}_{\mu} + \frac{\hbar\kappa}{2} \hat{\Lambda}_{\mu}, \qquad
	  \hat{\lambda}_{\mu} = \hat{\Lambda}_{\mu}, \qquad
	  \hat{\theta}^{\mu} = \hat{\Theta}^{\mu}, 
	  \label{quantum-classical-op}
\end{align}
where 
\begin{align}
	  \commut{ \hat{X}^{\mu} , \hat{P}_{\nu} } = 0, \qquad
	  \commut{ \hat{X}^{\mu} , \hat{\Lambda}_{\nu} } = i \delta^{\mu}_{\,\,\,\,\nu}, \qquad
	  \commut{ \hat{P}_{\mu} , \hat{\Theta}^{\nu} } = i \delta^{\nu}_{\,\,\,\,\mu}, \qquad
	  \commut{ \hat{\Lambda}_{\mu} , \hat{\Theta}^{\nu} } = 0.
\end{align} 
The generator $\boldsymbol{\mathcal{W}}$ can now be expressed in terms of the classical operators. The function $f$ is specified by requiring that 
the classical limit is recovered as $\kappa \rightarrow 0$. Hence, the quantum generator of dynamics in phase space reads 
\begin{align}
	  \boldsymbol{\mathcal{W}} =& \frac{1}{2m\kappa} \left(  \hat{P}_{\mu} + \frac{\hbar\kappa}{2} \hat{\Lambda}_{\mu}   \right) 
	    g^{\mu \nu}( \hat{X}-\frac{\hbar\kappa}{2} \hat{\Theta}) \left(  \hat{P}_{\nu} + \frac{\hbar\kappa}{2} \hat{\Lambda}_{\nu}  \right) 
	    + \frac{1}{\kappa}U\left(  \hat{X}^{\mu} - \frac{\hbar\kappa}{2} \hat{\Theta}^{\mu}  \right) \\
	  & - \frac{1}{2m\kappa} \left(  \hat{P}_{\mu} - \frac{\hbar\kappa}{2} \hat{\Lambda}_{\mu}   \right) g^{\mu \nu}(\hat{X}+\frac{\hbar\kappa}{2}\hat{\Theta}) 
	      \left(  \hat{P}_{\nu} - \frac{\hbar\kappa}{2} \hat{\Lambda}_{\nu}  \right) 
	      - \frac{1}{\kappa}U\left(  \hat{X}^{\mu} + \frac{\hbar\kappa}{2} \hat{\Theta}^{\mu}  \right),
	      \label{general-curvilinear-W}
\end{align}
which can be expanded as
\begin{align}
  \boldsymbol{\mathcal{W}} =& 
  \frac{ g^{\mu\nu}_{-} - g^{\mu\nu}_{+} }{2m\kappa} \hat{P}_{\mu} \hat{P}_{\nu} 
  + \frac{\hbar^2 \kappa}{8m} \hat{\Lambda}_{\mu}( g^{\mu\nu}_{-} - g^{\mu\nu}_{+} )\hat{\Lambda}_{\nu}
  + \frac{\hbar}{4m} (g^{\mu\nu}_{-} + g^{\mu\nu}_{+}) 
  \hat{P}_{\mu}\hat{\Lambda}_{\nu} 
  \\
 &  +
  \frac{\hbar}{4m} \hat{P}_{\mu}\hat{\Lambda}_{\nu}  (g^{\mu\nu}_{-} + g^{\mu\nu}_{+}) 
+\frac{1}{\kappa}\left[
   U\left(  \hat{X}^{\mu} - \frac{\hbar\kappa}{2} \hat{\Theta}^{\mu}  \right) 
 - U\left(  \hat{X}^{\mu} + \frac{\hbar\kappa}{2} \hat{\Theta}^{\mu}  \right)
\right]
\end{align}
with $g^{\mu\nu}_{+} 
 = g^{\mu\nu}( X + \frac{\hbar \kappa}{2} \Theta  )$ and 
$g^{\mu\nu}_{-} 
 = g^{\mu\nu}( X - \frac{\hbar \kappa}{2} \Theta  )$. 

In the $X\Theta$-representation
\begin{align}
	  \hat{X}^{\mu} = X^{\mu}, \qquad
	  \hat{\Lambda}_{\mu}  = -i \frac{\partial \,\,\,}{ \partial X^{\mu}}, \qquad
	  \hat{P}_{\mu} =  i \frac{\partial \,\,\,}{ \partial \Theta^{\mu}}, \qquad
	  \hat{\Theta}^{\mu} = \Theta^{\mu},
	   \label{XTheta-representation}
\end{align}
the generator of motion reads
\begin{align}
 \boldsymbol{\mathcal{W}} =& -\frac{g^{\mu\nu}_{-}-g^{\mu\nu}_{+}}{2m\kappa} 
 \frac{\partial^2}{\partial \Theta^{\mu} \Theta^{\nu}} 
 - \frac{\hbar^2 \kappa}{8 m}(g^{\mu\nu}_{-} - g^{\mu\nu}_{+})
   \frac{\partial^2}{\partial X^{\mu} X^{\nu}} 
    -  \frac{\hbar^2 \kappa}{8 m}
   \frac{\partial( g^{\mu\nu}_{-}-g^{\mu\nu}_{+}  )}{X^\mu} 
   \frac{\partial}{\partial X^\nu} 
   \\ &
    +\frac{\hbar (g^{\mu\nu}_{-} + g^{\mu\nu}_{+}) }{2m}
 \frac{\partial^2}{\partial \Theta^{\mu} X^{\nu}} 
  + \frac{\hbar}{4m}
    \frac{\partial(g^{\mu\nu}_{-} + g^{\mu\nu}_{+}  )}{\partial X^{\mu}}
     \frac{\partial}{\partial \Theta^{\nu}} 
   + \frac{1}{\kappa}\left[
    U\left(X-\frac{\hbar\kappa}{2} \Theta \right) - 
     U\left(X + \frac{\hbar \kappa}{2} \Theta \right)   
     \right].  
\end{align}

 The classical limit is readily calculated from Eq. (\ref{general-curvilinear-W})
\begin{align}
 \boldsymbol{\mathcal{W}}_{\kappa \rightarrow 0} = \hbar \left(
\frac{1}{m}g^{\alpha \mu} \hat{P}_{\mu} \hat{\Lambda}_{\alpha} 
- \frac{1}{2m}
 \hat{P}_{\mu} \hat{P}_{\nu} \frac{\partial g^{\mu \nu}}{\partial X^{\alpha}}  \hat{\Theta}^{\alpha} 
  -  \frac{\partial U}{\partial X^{\mu}} \hat{\Theta}^{\mu} 
  \right).
\end{align}

Applying the $XP$-representation,
\begin{align}
	  \hat{X}^{\mu} = X^{\mu}, \qquad
	  \hat{\Lambda}_{\mu}  = -i \frac{\partial \,\,\,}{ \partial X^{\mu}}, \qquad
	  \hat{P}_{\mu} =  P, \qquad
	  \hat{\Theta}^{\mu} = -i \frac{\partial \,\,\,}{ \partial P_{\mu}},
	   \label{XTheta-representation}
\end{align}
the classical Liouville equation in curvilinear coordinates
\begin{align}
 \frac{\partial \rho}{\partial t} =
 - \frac{1}{m} g^{\mu\nu} P_{\mu} \frac{\partial \rho }{X^{\nu}} + \frac{1}{2m}P_{\mu}P_{\nu} 
 \frac{\partial g^{\mu \nu}}{\partial X^{\alpha}} \frac{\partial \rho}{\partial P_{\alpha}}
 + \frac{\partial U}{\partial X
 ^{\mu}} \frac{\partial \rho}{\partial P_{\mu}},
\end{align}
corresponding to Hamiltonian equations (\ref{HE1}) and (\ref{HE2}), is finally obtained. 

\section{Classical Field Theory}\label{Sec_Classical_FieldTh}

In the current section, considering the classical Klein-Gordon field,  we demonstrate the utility of ODM for classical field theories. Conceptually, the case of classical field theories turns out to be similar to the single classical particle case.
(The Koopman-von Neumann approach has been extended to classical field theories in Refs. \cite{Carta2006, Gozzi2011, Cattaruzza2011}.)

The Ehrenfest theorems for the classical Klein-Gordon field read
\begin{align}
	& \frac{1}{c} \frac{d}{dt} \bra{\Psi(t)} \hat{\phi}(x) \ket{\Psi(t)} = \bra{\Psi(t)} \hat{\pi}(x) \ket{\Psi(t)}, \\
	& \frac{1}{c} \frac{d}{dt} \bra{\Psi(t)} \hat{\pi}(x) \ket{\Psi(t)} = \bra{\Psi(t)} \hat{\phi}'' (x) - \mu^2 \hat{\phi} (x) \ket{\Psi(t)},
		\qquad \mu = m c / \hbar,
\end{align}
where $x$ can be interpreted as a continuous index and
\begin{align}
	& \commut{ \hat{\phi}(x), \hat{\phi}(x') } = \commut{ \hat{\pi}(x), \hat{\pi}(x') } = \commut{ \hat{\phi}(x), \hat{\pi}(x') } = 0 \Longrightarrow \notag\\
	& \commut{ \hat{\phi}^{(n)}(x), \hat{\phi}^{(m)}(x') } = \commut{ \hat{\pi}^{(n)}(x), \hat{\pi}^{(m)}(x') } = 
		\commut{ \hat{\phi}^{(n)}(x), \hat{\pi}^{(m)}(x') } = 0, \quad \forall n,m \geqslant 0.
\end{align}
According to Stone's theorem (see Sec. \ref{Sec_Stones_Th}), we introduce the generator of motion, $\hat{L}$, of the state vector, $\ket{\Psi(t)}$,
\begin{align}
	i \ket{d\Psi(t)/dt} = c \hat{L} \ket{\Psi(t)};
\end{align}
hence,
\begin{align}
	i \commut{ \hat{L}, \hat{\phi}(x) } = \hat{\pi}(x), \qquad 
	i \commut{ \hat{L}, \hat{\pi}(x) } =  \hat{\phi}''(x) - \mu^2 \hat{\phi}(x) .
\end{align}

We shall seek $\hat{L}$ in the form
\begin{align}
	\hat{L} = \int dx' L \left( \hat{\lambda}_{\phi}(x'),  \hat{\lambda}_{\pi}(x'), \hat{\phi}(x'), \hat{\pi}(x'), \ldots, \hat{\phi}^{(n)}(x'), \hat{\pi}^{(n)} (x'), \ldots \right),
\end{align}
where the auxiliary operators $\hat{\lambda}_{\phi}(x)$ and $\hat{\lambda}_{\pi}(x)$ are assumed to obey
\begin{align}
	\commut{ \hat{\phi}(x), \hat{\lambda}_{\phi}(x') } = i\delta(x-x'), \qquad \commut{ \hat{\pi}(x), \hat{\lambda}_{\pi}(x') } = i\delta(x-x'),
	\qquad \commut{ \hat{\lambda}_{\phi}(x), \hat{\lambda}_{\pi}(x') } = 0.
\end{align}
Thus, employing theorem \ref{Th_Weyl_commutator_theorem} from Sec. \ref{Sec_Weyl_calculus}, we get
\begin{align}
	\hat{L} = \int dx' \left\{ \hat{\pi}(x') \hat{\lambda}_{\phi}(x') + \left[ \hat{\phi}''(x') -\mu^2\hat{\phi}(x') \right] \hat{\lambda}_{\pi}(x') 
			+ F\left(\ldots, \hat{\phi}^{(n)}(x'), \hat{\pi}^{(n)}(x'), \ldots \right) \right\},	
\end{align}
where $F=F(\ldots, \hat{\phi}^{(n)}(x), \hat{\pi}^{(n)}(x), \ldots)$ denotes an arbitrary real functional of derivatives of $\hat{\phi}(x)$ and $\hat{\pi}(x)$.

We shall find the equation of motion for the quantity $|\langle \phi(x) \, \pi(x) \ket{\Psi(t)}|^2$ -- the probability density for a Klein-Gordon field's state being given by $\phi(x)$ and $\pi(x)$ at time moment $t$. We introduce the notation
\begin{align}
	& \hat{\phi}(x) \ket{ \phi(x') \, \pi(x') } = \phi(x) \ket{ \phi(x') \, \pi(x') }, \qquad \hat{\pi}(x) \ket{ \phi(x') \, \pi(x') } = \pi(x) \ket{ \phi(x') \, \pi(x') } \Longrightarrow \notag\\
	& \hat{\phi}^{(n)}(x) \ket{ \phi(x') \, \pi(x') } = \phi^{(n)}(x) \ket{ \phi(x') \, \pi(x') }, \qquad \hat{\pi}^{(n)}(x) \ket{ \phi(x') \, \pi(x') } = \pi^{(n)}(x) \ket{ \phi(x') \, \pi(x') }.
\end{align}
The functional derivative of $F[f(x)]$ is defined as
\begin{align}
	\frac{ \delta F[f(x)]}{\delta f(x')} \coloneqq \lim_{\varepsilon \to 0} \frac{ F[f(x) + \varepsilon \delta(x-x')] - F[f(x)] }{\varepsilon}.
\end{align}
Whence,
\begin{align}
	\frac{ \delta }{\delta f(x')} \left\{ F[f(x)] G[f(x)] \right\} = G[f(x)] \frac{ \delta F[f(x)] }{\delta f(x')} + F[f(x)] \frac{ \delta G[f(x)] }{\delta f(x')}.
\end{align}
The operators $\hat{\lambda}_{\phi}(x)$, $\hat{\lambda}_{\pi}(x)$, $\hat{\phi}(x)$, and $\hat{\pi}(x)$ in the $\phi \pi$-representation read
\begin{align}
	\hat{\lambda}_{\phi}(x) = -i \frac{\delta}{\delta \phi(x)} + G\left(\ldots, \hat{\phi}^{(n)}(x'), \hat{\pi}^{(n)}(x'), \ldots \right), \qquad 
	\hat{\lambda}_{\pi}(x) = -i \frac{\delta}{\delta \pi(x)}, \qquad \hat{\phi}(x) = \phi(x), \qquad \hat{\pi}(x) = \pi(x),
\end{align}
where $G$ is any real functional.

\begin{align}
	\frac{1}{c} \frac{\partial}{\partial t} \langle \phi(x) \, \pi(x) \ket{\Psi(t)} =
		\int dx' \left\{ -\pi(x')\frac{\delta}{\delta \phi(x')} - \left[ \phi''(x') - \mu^2 \phi(x') \right] \frac{\delta}{\delta \pi(x')} 
		-iF \right\} \langle \phi(x) \, \pi(x) \ket{\Psi(t)}.
\end{align}
Here $F$ has ``absorbed'' $G$.
\begin{align}
	\frac{1}{c} \frac{\partial}{\partial t} \left| \langle \phi(x) \, \pi(x) \ket{\Psi(t)} \right|^2 =
		\int dx' \left\{ -\pi(x')\frac{\delta}{\delta \phi(x')} - \left[ \phi''(x') - \mu^2 \phi(x') \right] \frac{\delta}{\delta \pi(x')} \right\} 
		\left| \langle \phi(x) \, \pi(x) \ket{\Psi(t)} \right|^2 .
\end{align}
This equation is a first order functional partial differential equation. We employ the continuous analogue of the method of characteristics to get
\begin{align}
	c \delta t = \frac{\delta \phi(x)}{\pi(x)} = \frac{\delta \pi(x)}{\phi''(x) - \mu^2 \phi(x)} \Longrightarrow 
	\frac{1}{c} \frac{\delta \phi(x)}{\delta t} = \pi(x), \quad \frac{1}{c} \frac{\delta \pi(x)}{\delta t} = \phi''(x) - \mu^2 \phi(x).
\end{align}
These equations coincide with the classical Klein-Gordon equation 
\begin{align}\label{Classical_KleinGordon_Eq}
	\left[ \frac{1}{c^2} \frac{\partial^2}{\partial t^2} - \frac{\partial^2}{\partial x^2} + \mu^2 \right]\phi(x,t) = 0.
\end{align}

\section{Quantum Field Theory}\label{Sec_Quantum_FieldTh}

Now we shall employ ODM to perform the Bose and Fermi second quantization of the Schr\"{o}dinger equation. The application of ODM to other quantum field theoretic models should be straightforward. 

As before, we start from the Ehrenfest theorems 
\begin{align}
	& \frac{d}{dt} \bra{\Psi(t)} \hat{\psi}(x) \ket{\Psi(t)} = 
		\bra{\Psi(t)} -\frac{i}{\hbar} U(x) \hat{\psi}(x) + \frac{i \hbar}{2m} \frac{\partial^2 \hat{\psi}(x)}{\partial x^2} \ket{\Psi(t)}, \\
	& \frac{d}{dt} \bra{\Psi(t)} \hat{\psi}^{\dagger}(x) \ket{\Psi(t)} = 
		\bra{\Psi(t)} \frac{i}{\hbar} U(x) \hat{\psi}^{\dagger}(x) - \frac{i \hbar}{2m} \frac{\partial^2 \hat{\psi}^{\dagger}(x)}{\partial x^2} \ket{\Psi(t)}.
\end{align}
Note that the operator $\hat{\psi}(x)$ is not self-adjoint; thus, these Ehrenfest theorems are for complex quantities $\hat{\psi}(x)$ and $\hat{\psi}^{\dagger}(x)$.

Having introduced the generator of motion $\hat{H}$ by means of Stone's theorem (see Sec. \ref{Sec_Stones_Th})
\begin{align}
	i\hbar \ket{d \Psi(t)/dt} = \hat{H} \ket{\Psi(t)},
\end{align}
with $\ket{\Psi(t)}$ being an element of a Fock space, we find
\begin{align}\label{SchrodingerEq_commutator_Eqs}
	\commut{ \hat{H}, \hat{\psi}(x) } = -U(x)\hat{\psi}(x) + \frac{\hbar^2}{2m} \frac{\partial^2 \hat{\psi}(x) }{ \partial x^2}, \qquad
	\commut{ \hat{H}, \hat{\psi}^{\dagger}(x) } = U(x)\hat{\psi}^{\dagger}(x) - \frac{\hbar^2}{2m} \frac{\partial^2 \hat{\psi}^{\dagger}(x)}{ \partial x^2}.
\end{align}

\subsection{Bose Quantization}

In the Bose case, we postulate the following commutation relations:
\begin{align}
	\commut{ \hat{\psi}(x), \hat{\psi}(x') } = \commut{ \hat{\psi}^{\dagger}(x), \hat{\psi}^{\dagger}(x') } = 0, \qquad
		\commut{ \hat{\psi}(x), \hat{\psi}^{\dagger}(x') } = \delta(x-x');
\end{align}
whence,
\begin{align}
	& \Commut{ \frac{\partial \hat{\psi}^{\dagger}(x')}{\partial x'}, \hat{\psi}(x) } \coloneqq 
		\lim_{\varepsilon\to 0} \frac{ \commut{\hat{\psi}^{\dagger}(x'+\varepsilon), \hat{\psi}(x)} - \commut{\hat{\psi}^{\dagger}(x'), \hat{\psi}(x)}}{\varepsilon}
		= -\delta' (x'-x), \qquad \Commut{ \frac{\partial \hat{\psi}(x')}{\partial x'}, \hat{\psi}^{\dagger}(x) } = \delta' (x'-x), \notag\\ 
	& \Commut{ \frac{\partial \hat{\psi}(x)}{\partial x}, \frac{\partial \hat{\psi}(x')}{\partial x'} } =  
		\Commut{ \frac{\partial \hat{\psi}^{\dagger}(x)}{\partial x}, \frac{\partial \hat{\psi}^{\dagger}(x')}{\partial x} } = 0. \label{derivative_commutators_QFT}
\end{align}
and seek the generator of dynamics in the form
\begin{align}
	\hat{H} = \int dx' \mathrsfs{H} \left(x', \hat{\psi}(x'), \hat{\psi}^{\dagger}(x'), \frac{\partial \hat{\psi}(x')}{\partial x'}, \frac{\partial \hat{\psi}^{\dagger}(x')}{\partial x'} \right).
\end{align}
Using theorem \ref{Th_Weyl_commutator_theorem} from Sec. \ref{Sec_Weyl_calculus} and Eqs. (\ref{derivative_commutators_QFT}), we obtain
\begin{align}
	& -\mathrsfs{H}'_{\psi^{\dagger}} + \mathrsfs{H}''_{\frac{\partial\psi^{\dagger}}{\partial x}, \, x} 
		+ \mathrsfs{H}''_{\frac{\partial\psi^{\dagger}}{\partial x}, \, \psi} \frac{\partial \psi}{\partial x} 
		+  \mathrsfs{H}''_{\frac{\partial\psi^{\dagger}}{\partial x}, \, \psi^{\dagger}} \frac{\partial \psi^{\dagger}}{\partial x}
		+  \mathrsfs{H}''_{\frac{\partial\psi^{\dagger}}{\partial x}, \, \frac{\partial \psi}{\partial x}} \frac{\partial^2 \psi}{\partial x^2}
		+  \mathrsfs{H}''_{\frac{\partial\psi^{\dagger}}{\partial x}, \, \frac{\partial \psi^{\dagger}}{\partial x}} \frac{\partial^2 \psi^{\dagger}}{\partial x^2}
		= -U(x)\psi(x) + \frac{\hbar^2}{2m} \frac{\partial^2 \psi}{\partial x^2}, \\
	& -\mathrsfs{H}'_{\psi} + \mathrsfs{H}''_{\frac{\partial\psi}{\partial x}, \, x}
		+ \mathrsfs{H}''_{\frac{\partial\psi}{\partial x}, \, \psi} \frac{\partial \psi}{\partial x} 
		+  \mathrsfs{H}''_{\frac{\partial\psi}{\partial x}, \, \psi^{\dagger}} \frac{\partial \psi^{\dagger}}{\partial x}
		+  \mathrsfs{H}''_{\frac{\partial\psi}{\partial x}, \, \frac{\partial \psi}{\partial x}} \frac{\partial^2 \psi}{\partial x^2}
		+  \mathrsfs{H}''_{\frac{\partial\psi}{\partial x}, \, \frac{\partial \psi^{\dagger}}{\partial x}} \frac{\partial^2 \psi^{\dagger}}{\partial x^2}
		= -U(x)\psi^{\dagger}(x) + \frac{\hbar^2}{2m} \frac{\partial^2 \psi^{\dagger}}{\partial x^2}.
\end{align}
A solution of these equations is 
\begin{align}
	\mathrsfs{H} = \frac{\hbar^2}{2m} \frac{\partial \psi^{\dagger}}{\partial x}\frac{\partial \psi}{\partial x}
		+  C(x) \psi^{\dagger} \frac{\partial \psi}{\partial x} 
		+  C(x) \frac{\partial \psi^{\dagger}}{\partial x} \psi + \left[ U(x) + \frac{\partial C(x)}{\partial x}\right] \psi^{\dagger} \psi,
\end{align}
where $C(x)$ is any real function. Finally, the generator of motion reads 
\begin{align}
	\hat{H} =& \int dx' \left[ \frac{\hbar^2}{2m} \frac{ \partial \hat{\psi}^{\dagger}(x')}{\partial x'} \frac{ \partial \hat{\psi}(x')}{\partial x'}
		+ U(x') \hat{\psi}^{\dagger}(x') \hat{\psi}(x') \right],
\end{align}
where the term that can be represented as the total derivative under the integral was discarded.

\subsection{Fermi Quantization}\label{Sec_FermiQuantization}

In the Fermi case, the following anti-commutation conditions should be used
\begin{align}
	\acommut{ \hat{\psi}(x), \hat{\psi}(x') } = \acommut{ \hat{\psi}^{\dagger}(x), \hat{\psi}^{\dagger}(x') } = 0, \qquad
		\acommut{ \hat{\psi}(x), \hat{\psi}^{\dagger}(x') } = \delta(x-x');
\end{align}
whence,
\begin{align}
	& \Acommut{ \frac{\partial \hat{\psi}(x')}{\partial x'}, \hat{\psi}(x) } = \Acommut{ \frac{\partial \hat{\psi}(x')}{\partial x'}, \frac{\partial \hat{\psi}(x)}{\partial x} } 
		= \Acommut{ \frac{\partial \hat{\psi}^{\dagger}(x')}{\partial x'}, \hat{\psi}^{\dagger}(x) } 
		= \Acommut{ \frac{\partial \hat{\psi}^{\dagger}(x')}{\partial x'}, \frac{\partial \hat{\psi}^{\dagger}(x)}{\partial x} } = 0, \notag\\
	& \Acommut{ \frac{\partial \hat{\psi}^{\dagger}(x')}{\partial x'}, \hat{\psi}(x) } 
		= \Acommut{ \frac{\partial \hat{\psi}(x')}{\partial x'}, \hat{\psi}^{\dagger}(x) } = \delta' (x'-x).
\end{align}
The crucial difference between the Fermi and Bose cases is the following:
\begin{align}
	\left[ \hat{\psi}(x) \right]^n = \left[ \hat{\psi}^{\dagger} (x) \right]^n = \left[ \frac{\partial \hat{\psi}(x)}{\partial x} \right]^n 
	= \left[ \frac{\partial \hat{\psi}^{\dagger}(x)}{\partial x} \right]^n  = 0, \qquad \forall n \geqslant 2,
\end{align}
i.e., any power series of anti-commuting variables terminates after the linear term. Therefore, the generator of motion should be sought in the form
\begin{align}\label{FermiQuantization_SchEq_DefHamiltonian}
	\hat{H} =& \int dx' \left[ a_1(x') \hat{\psi}^{\dagger}(x') \hat{\psi}(x') + a_2(x') \frac{ \partial \hat{\psi}^{\dagger}(x')}{\partial x'} \hat{\psi}(x')  
		+ a_2^*(x') \hat{\psi}^{\dagger}(x') \frac{ \partial \hat{\psi}(x')}{\partial x'} 
		+  a_3(x') \frac{ \partial \hat{\psi}^{\dagger}(x')}{\partial x'} \frac{ \partial \hat{\psi}(x')}{\partial x'} \right].
\end{align}

Theorem \ref{Th_Weyl_commutator_theorem} is not convenient in the Fermi case because it is solely based on commutation relations. However, the following identity is a more suitable tool in the case of anti-commuting operators:
\begin{align}
	\commut{ \hat{A}_1 \cdots \hat{A}_{2n}, \hat{B} } &= \sum_{k=1}^{2n} (-1)^k \hat{A}_1 \cdots \hat{A}_{k-1} \acommut{ \hat{A}_k, 		\hat{B} } \hat{A}_{k+1} \cdots \hat{A}_{2n}. \label{commutator_A2N_V2}
\end{align}
We shall prove this equation by induction. Assuming that Eq. (\ref{commutator_A2N_V2}) is correct and utilizing
\begin{align} 
	\commut{ \hat{B}\hat{C}, \hat{A} } =  \commut{\hat{B}, \hat{A}}\hat{C} + \hat{B}\commut{\hat{C}, \hat{A}}, \qquad
	\commut{ \hat{B}\hat{C}, \hat{A} } = -\acommut{\hat{B}, \hat{A}}\hat{C} + \hat{B}\acommut{\hat{C}, \hat{A}}, \notag
\end{align} 
we obtain 
\begin{align}
	\commut{ \hat{A}_1 \cdots \hat{A}_{2n} \hat{A}_{2n+1} \hat{A}_{2n+2}, \hat{B} } =& 
	\commut{ \hat{A}\cdots\hat{A}_{2n}, \hat{B} } \hat{A}_{2n+1}\hat{A}_{2n+2} +  \hat{A}_1\cdots\hat{A}_{2n} \commut{\hat{A}_{2n+1}\hat{A}_{2n+2}, \hat{B}} \nonumber\\
	=& \sum_{k=1}^{2n} (-1)^k \hat{A}_1 \cdots \acommut{\hat{A}_k, \hat{B}} \cdots \hat{A}_{2n}\hat{A}_{2n+1}\hat{A}_{2n+2} \nonumber\\
	&	-\hat{A}_1 \cdots \hat{A}_{2n} \acommut{\hat{A}_{2n+1}, \hat{B}} \hat{A}_{2n+2}  + \hat{A}_1 \cdots \hat{A}_{2n}\hat{A}_{2n+1} \acommut{\hat{A}_{2n+2}, \hat{B}} \nonumber\\
	=& \sum_{k=1}^{2n+2} (-1)^k \hat{A}_1 \cdots \acommut{\hat{A}_k, \hat{B}} \cdots \hat{A}_{2n+2}. \nonumber
\end{align}
Hence, Eq. (\ref{commutator_A2N_V2}) is verified.

Now substituting Eq. (\ref{FermiQuantization_SchEq_DefHamiltonian}) into Eqs. (\ref{SchrodingerEq_commutator_Eqs}) and utilizing Eq. (\ref{commutator_A2N_V2}), we find 
\begin{align}
	& \left[ \frac{\partial a_2(x)}{\partial x} - a_1(x) \right]\hat{\psi}(x) + \left[ a_2(x) - a_2^*(x) + \frac{\partial a_3(x)}{\partial x} \right] \frac{\partial \hat{\psi}(x)}{\partial x} 
		+ a_3(x) \frac{\partial^2 \hat{\psi}(x)}{\partial x^2} = -U(x)\hat{\psi}(x) + \frac{\hbar^2}{2m} \frac{\partial^2 \hat{\psi}(x) }{ \partial x^2}, \\
	& \left[ a_1(x) - \frac{\partial a_2^*(x)}{\partial x} \right]\hat{\psi}^{\dagger}(x) + 
		\left[ a_2(x) - a_2^*(x) - \frac{\partial a_3(x)}{\partial x} \right] \frac{\partial \hat{\psi}^{\dagger}(x)}{\partial x}
		- a_3(x) \frac{\partial^2 \hat{\psi}^{\dagger}(x)}{\partial x^2}
		= U(x)\hat{\psi}^{\dagger}(x) - \frac{\hbar^2}{2m} \frac{\partial^2 \hat{\psi}^{\dagger}(x)}{ \partial x^2}.
\end{align}
Thus, the generator of motion is of the form
\begin{align}
	\hat{H} =& \int dx' \left[ \frac{\hbar^2}{2m} \frac{ \partial \hat{\psi}^{\dagger}(x')}{\partial x'} \frac{ \partial \hat{\psi}(x')}{\partial x'}
		+ U(x') \hat{\psi}^{\dagger}(x') \hat{\psi}(x') \right],
\end{align}
where the term that can be represented as the total derivative under the integral was discarded.

\end{widetext}

\bibliography{literature}
\end{document}